\documentclass[11pt]{amsart}
\usepackage[leqno]{amsmath}
\usepackage{epsfig}
\usepackage{xspace}
\usepackage{epsfig}
\usepackage{color}
\usepackage{amsmath}
\usepackage{amssymb}
\usepackage{enumerate}

\renewenvironment{proof}{
\noindent{\it
Proof.}}{\hfill$\Box$\vspace{0.1cm}}
\usepackage[cyr]{aeguill}
\usepackage{epsfig}
\usepackage{amsmath}
\usepackage{amssymb}
\usepackage{graphicx,color, algorithmic, algorithm}
\usepackage{color}
\usepackage{verbatim}
\newtheorem{lemma}{Lemma}
\newtheorem{proposition}{Proposition}
\newtheorem{theorem}{Theorem}
\newtheorem{corollary}{Corollary}
\newtheorem{claim}{Claim}

\theoremstyle{definition}

\usepackage{color}
\newcommand{\commentout}[1]{}
\usepackage{graphicx}
\newcommand{\CW}{\ensuremath{\mathcal{C}{\mathcal W}{\mathcal
      F}{\mathcal R}}\xspace}
\newcommand{\cR}{\ensuremath{\mathcal R}\xspace}
\newcommand{\cC}{\ensuremath{\mathcal C}\xspace}

\newcommand{\CWW}{\ensuremath{\mathcal{C}{\mathcal W}{\mathcal
      W}}\xspace}

\newcommand{\CWRC}{\ensuremath{\mathcal{CWRC}}\xspace}

\newcommand{\CWFRW}{\ensuremath{\mathcal{C}{\mathcal {WFR}}{\mathcal
      W}}\xspace}

\newcommand{\CWFR}{\ensuremath{\mathcal{C}{\mathcal {WFR}}}\xspace}

\begin{document}

\thispagestyle{empty}

\centerline{\Large\bf Cop and robber games when the robber can hide and ride}

\vspace{6mm}

\centerline{ {\sc J\'er\'emie Chalopin$^{\small 1}$,} {\sc Victor Chepoi$^{\small 1}$,}
{\sc Nicolas Nisse$^{\small 2}$,} and {\sc Yann Vax\`es$^{\small 1}$}}

\vspace{3mm}
\centerline{$^{1}$ Laboratoire d'Informatique Fondamentale de Marseille,}
\centerline{Universit\'e d'Aix-Marseille and CNRS,}
\centerline{CMI and Facult\'e des Sciences de Luminy,}
\centerline{F-13288 Marseille Cedex 9, France}

\medskip
\centerline{$^2$ MASCOTTE, INRIA, I3S, CNRS,}
\centerline{UNS, Sophia Antipolis, France}

\vspace{15mm} {\small{ \noindent {\bf Abstract.}  In the classical cop
    and robber game, two players, the cop $\mathcal C$ and the robber
    $\mathcal R$, move alternatively along edges of a finite graph
    $G=(V,E)$. The cop captures the robber if both players are on the
    same vertex at the same moment of time. A graph $G$ is called {\it
      cop win} if the cop always captures the robber after a finite
    number of steps.  Nowakowski, Winkler (1983) and Quilliot (1983)
    characterized the cop-win graphs as graphs admitting a dismantling
    scheme.  In this paper, we characterize in a similar way the class $\CW(s,s')$ of
    cop-win graphs in the game in which the cop and the robber move at
    different speeds $s'$ and $s$, $s'\le s$. We also establish some connections
    between cop-win graphs for this game with $s'<s$ and
    Gromov's hyperbolicity. In the particular case $s'=1$ and $s=2$, we
    prove that the class of cop-win graphs is exactly the well-known
    class of dually chordal graphs. We show that all classes $\CW(s,1),$ $s\ge 3$,
    coincide  and we provide a structural
    characterization of these graphs. We also investigate several
    dismantling schemes necessary or sufficient for the cop-win graphs
    in the game in which the robber is visible only every $k$ moves
    for a fixed integer $k>1$. We characterize the graphs which are
    cop-win for any value of $k$. Finally, we consider the game where
    the cop wins if he is at distance at most $1$ from the robber and
    we characterize via a specific dismantling scheme the bipartite graphs
    where a single cop wins in
    this game.}

}


\section{Introduction}

\subsection{The cop and robber game(s)}
The cop and robber game originated in the 1980's with the work of
Nowakowski, Winkler \cite{NowWin}, Quilliot \cite{Qui83}, and Aigner,
Fromme \cite{AigFro}, and since then has been intensively investigated
by numerous authors and under different names (e.g., hunter and rabbit
game \cite{IsKaKh}). Cop and robber is a pursuit-evasion game played
on finite undirected graphs.  Player cop $\mathcal C$ has one or
several cops who attempt to capture the robber $\mathcal R$. At the
beginning of the game, $\mathcal C$ occupies vertices for the initial
position of his cops, then $\mathcal R$ occupies another vertex.
Thereafter, the two sides move alternatively, starting with $\mathcal
C,$ where a move is to slide along an edge or to stay at the same
vertex, i.e. pass. Both players have full knowledge of the current
positions of their adversaries. The objective of $\mathcal C$ is to
capture $\mathcal R$, i.e., to be at some moment of time, or {\it
  step}, at the same vertex as the robber.  The objective of $\mathcal
R$ is to continue evading the cop.  A {\it cop-win graph}
\cite{AigFro,NowWin,Qui83} is a graph in which a \emph{single cop}
captures the robber after a finite number of moves for all possible
initial positions of $\mathcal C$ and $\mathcal R.$ Denote by
${\mathcal C}{\mathcal W}$ the set of all cop-win graphs.  The
cop-number of a graph $G$, introduced by Aigner and Fromme
\cite{AigFro}, is the minimum number of cops necessary to capture the
robber in $G$.
Different combinatorial (lower and upper) bounds on the cop number for different classes of graphs
were given in~\cite{AigFro,And86,BerInt,Chi,Fra87,Qui85,Schr,Theis}
(see also the survey paper \cite{Alspach} and the annotated bibliography \cite{FoTi}).


In this paper, we investigate the cop-win graphs for
three basic variants of the classical cop and robber game (for
continuous analogous of these games, see \cite{FoTi}). In the {\it cop and fast robber game},
introduced by Fomin, Golovach, and Kratochvil \cite{FomGolKra} and
further investigated in \cite{NisseSuchan} (see also
\cite{FoGoKrNiSu}), the cop is moving at
unit speed while the speed of the robber is an integer $s\ge 1$ or is
unbounded ($s\in {\mathbb N}\cup\{ \infty\}),$ i.e., at his turn,
$\mathcal R$ moves along a path of length at most $s$ which does not
contain vertices occupied by $\mathcal C$. Let ${\mathcal
  C}{\mathcal W}{\mathcal F}{\mathcal R}(s)$ denote the class of all graphs
in which a single cop having speed 1 captures a robber having speed $s.$ Obviously,
${\mathcal C}{\mathcal W}{\mathcal F}{\mathcal R}(1)={\mathcal C}{\mathcal W}.$
In a more general version,   we will suppose that
$\mathcal R$ moves with speed $s$ and $\mathcal C$ moves with speed
$s'\le s$ (if $s'>s$, then the cop can always capture the robber by
strictly decreasing at each move his distance to the robber). We will
denote the class of cop-win graphs for this version of the game by
${\mathcal C}{\mathcal W}{\mathcal F}{\mathcal R}(s,s')$. A {\it
  witness version} of the cop and robber game was recently introduced
by Clarke \cite{Clarke}. In this game, the robber has unit speed and
moves by having perfect information about cop positions.  On the
other hand, the cop no longer has full information about robber's
position but receives it only occasionally, say every $k$ units of
time, in which case, we say that $\mathcal R$ is {\it visible} to
$\mathcal C$, otherwise, $\mathcal R$ is {\it invisible} (this kind of
constraint occurs, for instance, in the ``Scotland Yard'' game
\cite{ScotlandYard}). Following \cite{Clarke}, we call a graph $G$
$k$-{\it winnable} if a single cop can guarantee a win with such
witness information and denote by ${\mathcal C}{\mathcal W}{\mathcal
  W}(k)$ the class of all $k$-winnable graphs. Notice that
${\mathcal C}{\mathcal W}{\mathcal F}{\mathcal R}(s)\subseteq {\mathcal C}{\mathcal W}{\mathcal W}(s)$ because the first
game can be viewed as a particular version of the second game in which  $\mathcal C$ moves
only at the turns when he receives  the information about $\mathcal R$. Finally, the game of {\it distance $k$ cop and robber}
introduced by Bonato and Chiniforooshan \cite{BonChi} is played in the same way as classical cop and robber,
except that the cop wins if a cop is within distance at most $k$ from the robber (following the name of an analogous
game in continuous spaces \cite{FoTi},  we will refer to this game as {\it cop and robber with radius of capture $k$}).
We denote by ${\mathcal C}{\mathcal W}{\mathcal R}{\mathcal C}(k)$ the set of all cop-win graphs in this game.

\subsection{Cop-win graphs}

Cop-win graphs (in ${\mathcal C}{\mathcal W}$) have been characterized by Nowakowski and Winkler
\cite{NowWin}, and Quillot \cite{Qui85} (see also \cite{AigFro}) as
dismantlable graphs (see Section \ref{sec:prelim} for formal
definitions).  Let $G=(V,E)$ be a graph and $u,v$ two vertices of $G$
such that any neighbor of $v$ (including $v$ itself) is also a
neighbor of $u$.  Then there is a retraction of $G$ to $G\setminus\{
v\}$ taking $v$ to $u$. Following \cite{HeNe}, we call this retraction
a {\it fold} and we say that $v$ is {\it dominated} by $u.$ A graph
$G$ is {\it dismantlable} if it can be reduced, by a sequence of
folds, to a single vertex.  In other words, an $n$-vertex graph
$G$ is dismantlable if its vertices can be ordered
$v_1,\ldots,v_n$ so that for each vertex $v_i, 1\le i<n,$ there exists
another vertex $v_j$ with $j>i,$ such that $N_1(v_i)\cap X_i\subseteq
N_1(v_j),$ where $X_i:=\{ v_i,v_{i+1},\ldots,v_n\}$ and
$N_1(v)$ denotes the closed neighborhood of $v.$ For a simple proof
that dismantlable graphs are the cop-win graphs, see the book
\cite{HeNe}.  An alternative (more algorithmic) proof of this result
is given in \cite{IsKaKh}. Dismantlable graphs include bridged graphs
(graphs in which all isometric cycles have length 3) and Helly graphs
(absolute retracts) \cite{BaCh_survey,HeNe} which occur in several
other contexts in discrete mathematics. Except the cop and robber
game, dismantlable graphs are used to model physical processes like
phase transition \cite{BriWin}, while bridged graphs occur as
1-skeletons of systolic complexes in the intrinsic geometry of
simplicial complexes \cite{Ch_CAT,Hag,JanSwi}. Dismantlable graphs are
closed under retracts and direct products, i.e., they constitute a
variety \cite{NowWin}.

\subsection{Our results}

In this paper, we characterize the graphs of the class $\CW(s,s')$ for
all speeds $s,s'$ in the same vein as cop-win graphs, by using a
specific dismantling order. Our characterization allows to decide in
polynomial time if a graph $G$ belongs to any of considered classes
$\CW(s,s')$.  In the particular case $s' = 1$, we show that $\CW(2)$
is exactly the well-known class of dually chordal graphs.  Then we
show that the classes $\CW(s)$ coincide  for all $s \geq 3$ and that
the graphs $G$ of these classes have the following structure: the
block-decomposition of $G$ can be rooted in such a way that any block
has a dominating vertex and that for each non-root block, this
dominating vertex can be chosen to be the articulation point
separating the block from the root. We also establish some connections
between the graphs of $\CW(s,s')$ with $s'<s$ and  Gromov's hyperbolicity.
More precisely, we prove that any $\delta$-hyperbolic
graph belongs to the class $\CW(2r,r+2\delta)$ for any $r>0$, and that, for any
$s\geq 2s'$, the graphs in $\CW(s,s')$ are $(s-1)$-hyperbolic. We also establish
that Helly graphs and bridged graphs belonging to $\CW(s,s')$ are $s^2$-hyperbolic and
we conjecture that, in fact all graphs of $\CW(s,s'),$ where $s'<s,$ are
$\delta$-hyperbolic, where $\delta$ depends only of $s.$

In the second part of our paper, we
characterize the graphs that are $s$-winnable for all $s$ (i.e.,
graphs in $\cap_{s \geq 1}\CWW(s)$) using a similar decomposition
as for the graphs from the classes $\CW(s),$ $s\ge 3.$  On
the other hand, we show that for each $s$, $\CWW(s)\setminus\CWW(s+1)$
is non-empty , contrary to the classes $\CW(s)$. We show that all
graphs of $\CWW(2)$, i.e., the $2$-winnable graphs, have a special
dismantling order (called bidismantling), which however does not
ensure that a graph belongs to $\CWW(2).$ We present a stronger
version of bidismantling and show that it is sufficient for ensuring
that a graph is 2-winnable.  We extend bidismantling to any $k\ge 3$
and prove that for all odd $k,$ bidismantling is sufficient to ensure
that $G\in \CW(k).$ Finally, we characterize the bipartite members of
${\mathcal C}{\mathcal W}{\mathcal R}{\mathcal C}(1)$ via an
appropriate dismantling scheme. We also formulate several open
questions.


\subsection{Preliminaries} \label{sec:prelim}
For a graph $G=(V,E)$ and a subset $X$ of its vertices, we denote by
$G(X)$ the subgraph of $G$ induced by $X.$ We will write $G\setminus
\{ x\}$ and $G\setminus \{ x,y\}$ instead of $G(V\setminus\{ x\})$ and
$G(V\setminus \{ x,y\}).$ The {\it distance} $d(u,v):=d_G(u,v)$
between two vertices $u$ and $v$ of a graph $G$ is the length (number
of edges) of a shortest $(u,v)$-path.  An induced subgraph $H$ of $G$
is {\it isometric} if the distance between any pair of vertices in $H$
is the same as that in $G.$ The {\it ball} (or disk) $N_r(x)$ of
center $x$ and radius $r\ge 0$ consists of all vertices of $G$ at
distance at most $r$ from $x.$ In particular, the unit ball $N_1(x)$
comprises $x$ and the neighborhood $N(x).$ The {\it punctured ball}
$N_r(x,G\setminus \{ y\})$ of center $x,$ radius $r,$ and puncture $y$
is the set of all vertices of $G$ which can be connected to $x$ by a
path of length at most $r$ avoiding the vertex $y,$ i.e., this is the
ball of radius $r$ centered at $x$ in the graph $G\setminus\{ y\}.$
A {\it retraction} $\varphi$ of a graph $H=(W,F)$ is
an idempotent nonexpansive mapping of $H$ into itself, that is,
$\varphi^2=\varphi:W\rightarrow W$ with $d(\varphi (x),\varphi (y))\le
d(x,y)$ for all $x,y\in W.$ The subgraph of $H$ induced by the image
of $H$ under $\varphi$ is referred to as a {\it retract} of $H.$

A \emph{strategy} for the cop is a function $\sigma$ which takes as
an input the  first $i$ moves of both players and outputs the 
$(i+1)$th move $c_{i+1}$ of the cop. A strategy  for the
robber is defined in a similar way. A cop's strategy $\sigma$ is
\emph{winning} if for any sequence of moves of the robber, the cop,
following $\sigma$, captures the robber after a finite sequence of
moves.  Note that if the cop has a winning strategy $\sigma$ in a
graph $G$, then there exists a winning strategy $\sigma'$ for the
cop that only depends  of the last positions of the two players (such a
strategy is called \emph{positional}).
This is because cop and robber games are parity games (by considering
the directed graph of configurations) and parity games always admit
positional strategies for the winning player \cite{Ku}.
A strategy for the cop is called \emph{parsimonious} if at his turn,
the cop captures the robber (in one move) whenever he can. For
example, in the cop and fast robber game, at his move, the cop
following a parsimonious strategy always captures a robber located at
distance at most $s'$ from his current position. It is easy to see
that in the games investigated in this paper, if the cop has a
(positional) winning strategy, then he also has a parsimonious
(positional) winning strategy.

\section{Cop-win graphs for game with fast robber: class $\CW(s,s')$}\label{sec-speed}

In this section, first we characterize the graphs of $\CW(s,s')$ via a specific
dismantling scheme, allowing to recognize them in polynomial time. Then we
show that any $\delta$-hyperbolic graph belongs to the class $\CW(2r,r+2\delta)$
for any $r\ge 1.$ We conjecture that the converse is true, i.e., any graph from
$\CW(s,s')$ with $s'<s$ is $\delta$-hyperbolic for some value of $\delta$ depending
only of $s,$ and we confirm this conjecture in several particular cases.

\subsection{Graphs of $\CW(s,s')$} For technical convenience, we will
consider a slightly more general version of the game:
given a subset of vertices $X$ of a graph $G=(V,E),$ the $X$-{\it
  restricted game} with cop and robber having speeds $s'$ and $s,$
respectively, is a variant in which $\mathcal C$ and $\mathcal R$ can
pass through any vertex of $G$ but can stand only at vertices of $X$
(i.e., the beginning and the end of each move are in $X$). A subset of vertices
$X$ of a graph $G=(V,E)$
is $(s,s')$-{\it winnable} if the cop captures the robber in the $X$-restricted
game. In the following, given a subset $X$ of admissible positions, we say
that a sequence of vertices $S_r = (a_1,\ldots, a_p, \ldots)$ of a
graph $G=(V,E)$ is $X$-{\it valid} for a robber with speed $s$
(respectively, for a cop with speed $s'$) if, for any $k$, we have $a_k
\in X$ and $d(a_{k-1},a_k)\le s$ (respectively, $d(a_{k-1},a_k)\leq
s'$).  We will say that a subset of vertices $X$ of a graph $G=(V,E)$
is $(s,s')$-{\it dismantlable} if the vertices of $X$ can be ordered
$v_1,\ldots,v_m$ in such a way that for each vertex $v_i, 1\le i<m,$
there exists another vertex $v_j$ with $j>i,$ such that
$N_{s}(v_i,G\setminus\{ v_j\})\cap X_i\subseteq N_{s'}(v_j),$ where
$X_i:=\{ v_i,v_{i+1},\ldots,v_m\}$ and $X_m=\{ v_m\}.$ A graph
$G=(V,E)$ is $(s,s')$-{\it dismantlable} if its vertex-set $V$ is
$(s,s')$-dismantlable.

\begin{theorem} \label{copwin} For any $s,s'\in {\mathbb N}\cup \{ \infty\},$ $s'\leq s$,
a graph $G=(V,E)$ belongs to the class  $\CW(s,s')$
if and only if $G$ is $(s,s')$-dismantlable.
\end{theorem}

\begin{proof}
First, suppose that $G$ is $(s,s')$-dismantlable and let $v_1,\ldots,
v_n$ be an $(s,s')$-dismantling ordering of $G.$ By induction on $n-i$
we will show that for each level-set $X_i=\{ v_i,\ldots,v_n\}$ the cop
captures the robber in the $X_i$-restricted game. This is obviously
true for $X_n=\{ v_n\}.$ Suppose that our assertion is true for all
sets $X_n,\ldots,X_{i+1}$ and we will show that it still holds for
$X_i$. Let $N_{s}(v_i,G\setminus\{ v_j\})\cap X_i\subseteq
N_{s'}(v_j)$ for a vertex $v_j\in X_i$. Consider a parsimonious
positional winning strategy $\sigma_{i+1}$ for the cop in the
$X_{i+1}$-restricted game.  We build a parsimonious winning strategy $\sigma_i$
for the cop in the $X_i$-restricted game: the intuitive idea is that
if the cop sees the robber in $v_i$, he plays as in the
$X_{i+1}$-restricted game when the robber is in $v_j$.  Let $\sigma_i$
be the strategy for the $X_i$-restricted game defined as follows. For
any positions $c \in X_i$ of the cop and $r \in X_i$ of the robber,
set $\sigma_{i}(c,r) = r$ if $d(c,r) \leq s'$, otherwise $\sigma_i(c,r) =
\sigma_{i+1}(c,r)$ if $c,r \neq v_i$, $\sigma_i(c,v_i) =
\sigma_{i+1}(c,v_j)$ if $c \notin \{ v_i, v_j\}$, and $\sigma_i(v_i,r)
= v_j$ if $r \neq v_i$ (in fact, if the cop plays $\sigma_i$ he
will never move to $v_i$ except to capture the robber there).  By construction,
the strategy $\sigma_{i}$ is parsimonious; in particular,
$\sigma_{i}(v_j,v_i) = v_i$, because $d(v_i,v_j) \leq s'$. We now prove that $\sigma_i$ is winning.

Consider any $X_i$-valid sequence $S_r=(r_1, \ldots, r_p, \ldots)$ of
moves of the robber and any trajectory $(\pi_1, \ldots, \pi_p,
\ldots)$ extending $S_r$, where $\pi_p$ is a simple path of length at most $s$ from $r_p$ to
$r_{p+1}$ along which the robber moves.  Let $S_r'=(r'_1, \ldots r'_p,
\ldots)$ be the sequence obtained by setting $r'_k=r_k$ if $r_k \neq
v_i$ and $r'_k=v_j$ if $r_k=v_i$. For each $p$, set $\pi_p' = \pi_p$
if $v_i\notin \{r_p,r_{p+1}\}$. If $v_i = r_{p+1}$ (resp. $v_i = r_p$
), set $\pi_p'$ be a shortest path from $r_p$ to $v_j$ (resp. from $v_j$ to
$r_{p+1}$) if $\pi_p$ does not contain $v_j$ and set $\pi_p'$ be the
subpath of $\pi_p$ between $r_p$ and $v_j$ (resp. between $v_j$ and
$r_{p+1}$) otherwise.  Since $N_{s}(v_i,G\setminus\{ v_j\})\cap X_i\subseteq
N_{s'}(v_j),$ we infer that $S_r'$ is a $X_{i+1}$-valid
sequence of moves for the robber. By induction hypothesis, for any initial location of $\mathcal C$ in
$X_{i+1}$, the strategy $\sigma_{i+1}$ allows the cop to capture the
robber which moves according to $S_r'$ in the $X_{i+1}$-restricted
game. Let $c_{m+1}'$ be the position of the cop after his last move
and $S_c' = (c_1', \ldots, c_{m+1}')$ be the sequence of
positions of the cop in the $X_{i+1}$-restricted game against $S_r'$
using $\sigma_{i+1}.$ Let $S_c = (c_1, \ldots, c_p, \ldots)$ be the
sequence of positions of the cop in the $X_i$-restricted game
against $S_r$ using $\sigma_i$. From the definition of $S'_r$ and
$\sigma_i,$ $S_c$ and $S'_c$ coincide at least until step $m,$
i.e., $c'_k=c_k$ for $k=1,\hdots,m.$ Moreover, if $c'_{m+1}\neq c_{m+1}$
then $c_{m+1} =r_m=v_i$ and $c'_{m+1} = r'_m= v_j.$
In the $X_{i+1}$-restricted version of the game, the robber is captured,
either (i) because after his last move, his position $r_m'$ is at distance
at most $s'$ from cop's current position $c'_{m}$, or (ii) because his
trajectory $\pi_m'$ from $r_m'$ to $r_{m+1}'$ passes via $c_{m+1}'$.

In case (i), since $d(r'_m,c'_m) \leq s'$ and the strategy $\sigma_{i+1}$
is parsimonious, we conclude that $c_{m+1}' = r'_m$.  If $c_{m+1}' = r'_m
\ne v_j$, then from the definition of $S'_r$ and $\sigma_i$, we
conclude that $c_{m+1} = c_{m+1}' = r'_m = r_m$, whence $c_{m+1}=r_m$
and $\mathcal C$ captures $\mathcal R$ using $\sigma_i$. Now suppose
that $c_{m+1}' = r'_m = v_j$.  If $r_m = v_j$, then $d(c_{m},r_m) \leq
s'$ because $c_m = c_m'$ and thus $\mathcal C$ captures $\mathcal R$ at
$v_j$ using $\sigma_i$. On the other hand, if $r_m = v_i$, either
$c_{m+1} = v_i$ and we are done, or $c_{m+1} = v_j$ and since
$N_{s}(v_i,G\setminus\{v_j\})\cap X_i\subseteq N_{s'}(v_j),$
the robber is captured at the next move of the cop, i.e., $c_{m+2} =
r_{m+1}$ holds.

In case (ii), either the path $\pi_m'$ from $r_m'$ to
$r_{m+1}'$ is a subpath of $\pi_m$, or $v_i \in \{r_m,r_{m+1}\}$ and
$\pi_m$ does not go via $v_j$. In the first case, note that $c_{m+1} =
c_{m+1}'$, otherwise $c_{m+1} = v_i = r_m$ by construction of
$\sigma_i$ and thus the robber has been captured before. Therefore the
trajectory $\pi_m$ of the robber in the $X_i$-game traverses the
position $c_{m+1}$ of the cop and we are done. Now suppose that $\pi_m$ does not go via $v_j$ and $v_i \in
\{r_m,r_{m+1}\}$.  Note that in this case, $c_{m+1} = c'_{m+1}$ holds;
otherwise, $c_{m+1}'=r'_m=v_j$ and $c_{m+1}=r_m=v_i$ and therefore, the robber
is caught at step $m+1$. If $c_{m+1}$ belongs to $\pi_m$, then we are done
as in the first case. So suppose that $c_{m+1} \notin \pi_m$.
If $r_{m+1} = v_i$, then $r_m \in N_s(v_i, G\setminus\{v_j\})
\subseteq N_{s'}(v_j)$. Since, $\pi_m'$ is a shortest path and
$c'_{m+1}$ belongs to this path, $d(c'_{m+1},v_j) \leq s'$ and thus
either $c_{m+2} = v_i = r_{m+1}$ if $d(c'_{m+1},v_i) \leq s'$, or
$c_{m+2} = v_j$ since $\sigma_{i+1}$ is parsimonious. In the latter case,
since $N_s(v_i, G\setminus\{v_j\}) \subseteq N_{s'}(v_j)$, $r_{m+1} =
v_i,$ and $c_{m+2} = v_j$, the robber will be captured at the next
move.
Finally, suppose that $r_m = v_i$. Then $r_m' = v_j$. Since $\pi_m$ is
a path of length at most $s$ avoiding $v_j$, we conclude that $r_{m+1}
\in N_s(v_i, G\setminus\{v_j\}) \subseteq N_{s'}(v_j)$. Since $\pi_m'$
is a shortest path from $v_j$ to $r_{m+1}$ containing the vertex
$c'_{m+1} = c_{m+1}$, we have $d(c_{m+1},r_{m+1}) \leq d(v_j,r_{m+1})
\leq s'$. Therefore, the cop captures the robber in $r_{m+1}$ at his
next move, i.e., $c_{m+2} = r_{m+1}$. This shows that a $(s,s')$-dismantlable
graph $G$ belongs to $\CW(s,s').$

Conversely, suppose that for a $X$-restricted game played on a graph
$G=(V,E)$ there is a positional winning strategy $\sigma$ for the
cop. We assert that $X$ is $(s,s')$-dismantlable. This is obviously
true if $X$ contains a vertex $y$ such that $d(y,x)\le s'$ for any
$x\in X.$ So suppose that $X$ does not contain such a vertex $y.$
Consider a $X$-valid sequence of moves of the robber having a
maximum number of steps before the capture of the robber. Let $u\in
X$ be the position occupied by the cop before the capture of
$\mathcal R$ and let $v\in X$ be the position of the robber at this
step.  Since wherever the robber moved next in $X$ (including
remaining in $v$ or passing via $u$), the cop would capture him,
necessarily $N_s(v,G\setminus\{ u\})\cap X\subseteq N_{s'}(u)$
holds. Set $X':=X\setminus \{ v\}$.

We assert that $X'$ is $(s,s')$-winnable as well.  In this proof, we
use a strategy that is not positional but uses one bit of memory.  A
strategy using one bit memory can be presented as follows: it is a
function which takes as input the current positions of the two players
and a boolean (the current value of the memory) and that outputs the
next position of the cop and a boolean (the new value of the
memory). Using the positional winning strategy $\sigma$, we define
$\sigma'(c,r,m)$ for any positions $c \in X'$ of the cop and $r \in
X'$ of the robber and for any value of the memory $m \in \{0,1\}$.
The intuitive idea for defining $\sigma'$ is that the cop plays using
$\sigma$ except when he is in $u$ and his memory contains $1$; in this
case, he uses $\sigma$ as if he was in $v$.  If $m = 0$ or $c \neq u$,
then we distinguish two cases: if $\sigma(c,r) = v$ then
$\sigma'(c,r,m) = (u,1)$ (this is a valid move since $N_{s'}(v) \cap X
\subseteq N_{s'}(u)$) and $\sigma'(c,r,m) = (\sigma(c,r),0)$
otherwise.  If $m = 1$ and $c = u$, we distinguish two cases: if
$\sigma(v,r) = v$, then $\sigma'(u,r,1) = (u,1)$ and $\sigma'(u,r,1) =
(\sigma(v,r),0)$ otherwise (this is a valid move since $N_{s'}(v)\cap
X \subseteq N_{s'}(u)$).  Let $S_r= (r_1, \ldots, r_p, \ldots)$
be any $X'$-valid sequence of moves of the robber. Since $X'\subset
X,$ $S_r$ is also a $X$-valid sequence of moves of the robber.  Let
$S_c:=(c_1, \ldots, c_p, \ldots)$ be the corresponding $X$-valid
sequence of moves of the cop following $\sigma$ against $S_r$ in $X$
and let $S'_c=(c'_1, \ldots, c'_p, \ldots)$ be the $X'$-valid sequence
of moves of the cop following $\sigma'$ against $S_r$. Note that the
sequences of moves $S_c$ and $S_c'$ differ only if $c_k = v$ and $c_k'
= u$.  Finally, since the cop follows a winning strategy for $X,$
there is a step $j$ such that $c_j=r_j \in X\setminus \{v\}$ (note
that $r_j\neq v$ because we supposed that $S_r\subseteq X'$).  Since
$c_j \neq v$, we also have $c'_j=r_j$, thus $\mathcal C$ captures
$\mathcal R$ in the $X'$-restricted game. Starting from a positional
strategy for the $X$-restricted game, we have constructed a winning
strategy using memory for the $X'$-restricted game. As mentioned in
the introduction, it implies that there exists a positional winning
strategy for the $X'$-restricted game.

Applying induction on the number of vertices of the cop-winning set
$X,$ we conclude that $X$ is $(s,s')$-dismantlable.  Applying this
assertion to the vertex set $V$ of cop-win graph $G=(V,E)$ from the
class $\CW(s,s'),$ we will conclude that $G$ is $(s,s')$-dismantlable.
\end{proof}

\begin{corollary} Given a graph $G=(V,E)$ and the integers  $s,s'\in {\mathbb N}\cup \{ \infty\},$ $s'\leq s,$
one can recognize in polynomial time if $G$ belongs to $\CW(s,s').$
\end{corollary}

\begin{proof} By Theorem \ref{copwin},  $G\in \CW(s,s')$
if and only if $G$ is $(s,s')$-dismantlable. Moreover, from the last part of the proof of
Theorem \ref{copwin} we conclude that if
a subset $X$ of vertices of $G$ is $(s,s')$-winnable and $u,v\in X$ such that
$N_s(v,G\setminus\{ u\})\cap X\subseteq N_{s'}(u)$ holds, then the set $X'=X\setminus \{ v\}$
is $(s,s')$-winnable as well. Therefore it suffices to run the following algorithm. Start with $X:=V$
and as long as possible find in $X$ two vertices $u,v$ satisfying the inclusion
$N_s(v,G\setminus\{ u\})\cap X\subseteq N_{s'}(u),$ and set $X:=X\setminus \{ v\}.$ If the algorithm ends
up with a set $X$ containing at least two vertices, then $G$ is not $(s,s')$-winnable, otherwise, if
$X$ contains a single vertex, then
$G$ is $(s,s')$-dismantlable and therefore $G\in \CW(s,s').$
\end{proof}

\subsection{Graphs of $\CW(s,s')$ and hyperbolicity}
Introduced by  Gromov \cite{Gr}, $\delta$-hyperbolicity of a metric space
measures, to some extent,  the deviation of a metric from a tree metric.
A graph $G$ is $\delta$-{\it hyperbolic} if for any four vertices
$u,v,x,y$ of $G$, the two larger of the three distance sums
$d(u,v)+d(x,y)$, $d(u,x)+d(v,y)$, $d(u,y)+d(v,x)$ differ by at most
$2\delta \geq 0$. Every 4-point metric $d$ has a canonical
representation in the rectilinear plane as illustrated in Fig. \ref{fig0}.
The three distance sums  are ordered from small to large, thus implying
$\xi \le \eta.$ Then $\eta$ is half the difference of the largest
and the smallest sum, while $\xi$ is half the largest minus the
medium sum. Hence, a graph $G$ is $\delta$-hyperbolic iff
$\xi$ does not exceed $\delta$ for any four vertices $u,v,w,x$ of $G.$  Many classes of graphs are known to have bounded
hyperbolicity  \cite{BaCh_survey,ChDrEsHaVa}.  Our next
result, based on Theorem~\ref{copwin} and a  result of
\cite{ChDrEsHaVa}, establishes that in a $\delta$-hyperbolic graph a
``slow'' cop captures a faster robber provided that
$s'>s/2+2\delta$ (in the same vein, Benjamini \cite{Benj} showed that
in the competition of two growing clusters in a $\delta$-hyperbolic
graph, one growing faster that the other, the faster cluster not
necessarily surround the slower cluster).

\begin{figure}
\begin{center}
{\input{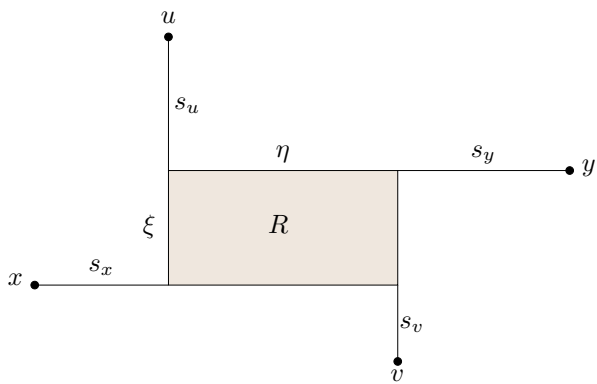}}
\end{center}
\caption{Realization of a 4-point metric in the rectilinear plane.}
\label{fig0}
\end{figure}

\begin{proposition}\label{prop:hyperbolic}
Given $r\ge 2\delta\ge 0,$ any $\delta$-hyperbolic graph $G=(V,E)$ is $(2r,r+2\delta)$-dismantlable
and therefore $G\in \CW(2r,r+2\delta).$
\end{proposition}

\begin{proof}
The second assertion follows from Theorem \ref{copwin}. To prove the
$(2r,r+2\delta)$-dismantlability of $G,$ we will employ Lemma 2 of
\cite{ChDrEsHaVa}.  According to this result, in a $\delta$-hyperbolic
graph $G$ for any subset of vertices $X$ there exist two vertices $x
\in X$ and $c\in V$ such that $d(c,y) \leq r+2\delta$ for any vertex
$y \in X \cap N_{2r}(x),$ i.e., $N_{2r}(x)\cap X\subseteq
N_{r+2\delta}(c).$ The proof of \cite{ChDrEsHaVa} shows that the
vertices $x$ and $c$ can be selected in the following way: pick any
vertex $z$ of $G$ as a basepoint, construct a breadth-first search
tree $T$ of $G$ rooted at $z,$ and then pick $x$ to be the furthest
from $z$ vertex of $X$ and $c$ to be vertex located at distance
$r+2\delta$ from $x$ on the unique path between $x$ and $z$ in $T.$
Using this result, we will establish a slightly stronger version of
dismantlability of a $\delta$-hyperbolic graph $G$, in which the
inclusion $N_{s}(v_i,G\setminus \{ v_j\})\cap X_i\subseteq
N_{s'}(v_j)$ is replaced by $N_{s}(v_i)\cap X_i\subseteq N_{s'}(v_j)$
with $s:=2r$ and $s':=r+2\delta$. We recursively construct the
ordering of $V$. By previous result, there exist two vertices $v_1 \in
X_1:=V$ and $c \in X_2:=V\setminus \{v_1\}$ such that $N_{2r}(v_1)\cap
X_1 \subseteq N_{r+2\delta}(c)$. At step $i \geq 1$, suppose by
induction hypothesis that $V$ is the disjoint union of the sets
$\{v_1,\ldots,v_i\}$ and $X_{i+1}$, so that, for any $j\leq i$, there
exists a vertex $c \in X_{j+1}$ such that $N_{2r}(v_j)\cap X_j
\subseteq N_{r+2\delta}(c)$ with $X_j=\{v_j,\ldots,v_i\} \cup
X_{i+1}$. We assert that this ordering can be extended. Applying the
previous result to the set $X:=X_{i+1}$ we can define two vertices
$v_{i+1} \in X_{i+1}$ and $c\ne v_{i+1}$ such that
$N_{2r}(v_{i+1})\cap X_{i+1} \subseteq N_{r+2\delta}(c)$. The choice
of the vertices $x\in X$ and $c\in V$ provided by \cite{ChDrEsHaVa}
and the definition of the sets $X_1,X_2,\ldots$ ensure that if a
vertex of $G$ is closer to the root than another vertex, then the
first vertex will be labeled later than the second one. Since by
construction $c$ is closer to $z$ than $v_{i+1},$ necessarily $c$
belongs to the set $X_{i+1}\setminus \{ v_{i+1}\}.$
\end{proof}

In general, dismantlable graphs do not have bounded hyperbolicity because they are universal in the following sense.
As we noticed in the introduction, any finite Helly graph is dismantlable. On the other hand, it is well known that
an arbitrary connected graph can be isometrically embedded into a Helly graph (see for example \cite{BaCh_survey,Qui83}). However,
dismantlable graphs without some short induced cycles are 1-hyperbolic:

\begin{corollary} Any dismantlable graph $G=(V,E)$ without induced 4-,5-, and 6-cycles is 1-hyperbolic, and
therefore $G\in \CW(2r,r+2)$ for any $r>0.$
\end{corollary}

\begin{proof} A dismantlable graph $G$ not containing induced 4- and
  5-cycles does not contain 4-wheels and 5-wheels as well
(a $k$-\emph{wheel} is a cycle of length $k$ plus a vertex adjacent to all
  vertices of this cycle), therefore $G$ is bridged by a result of
  \cite{AnFa}. Since $G$ does not contain 6-wheels as well, $G$ is
  1-hyperbolic by Proposition 11 of \cite{ChDrEsHaVa}.  Then the
  second assertion immediately follows from Proposition
  \ref{prop:hyperbolic}.
\end{proof}

\medskip\noindent
{\bf Open question 1:} Is it true that the converse of Proposition
\ref{prop:hyperbolic} holds? More precisely, is it true that if $G\in
\CW(s,s')$ for $s'<s,$ then the graph $G$ is $\delta$-hyperbolic,
where $\delta$ depends only of $s$?
\medskip

We give some confidences in the truth of this conjecture by showing
that for $s\ge 2s'$ all graphs $G\in \CW(s,s')$ are
$(s-1)$-hyperbolic.  On the other hand, since $\CW(s,s')\subset
\CW(s,s'+1),$ to answer our question for $s'<s<2s'$ it suffices to show its
truth for the particular case $s'=s-1.$ We give a positive answer to our question for Helly and bridged graphs by
showing that if such a graph $G$ belongs to
the class $\CW(s,s-1),$ then $G$ is  $s^2$-hyperbolic.

In the following results, for an $(s,s')$-dismantling order
$v_1,\ldots,v_n$ of a graph $G\in \CW(s,s')$ and a vertex $v$ of
$G$, we will denote by  $\alpha(v)$ the rank of $v$ in this order
(i.e., $\alpha(v)=i$ if $v=v_i$). For two vertices $u,v$ with
$\alpha(u)<\alpha(v)$ and a shortest $(u,v)$-path $P(u,v),$ an
$s$-{\it net} $N(u,v)$ of $P(u,v)$ is an ordered subset $(u=x_0,x_1,\ldots,x_k,x_{k+1}=v)$
of vertices of $P(u,v)$, such that $d(x_i,x_{i+1})=s$ for
any $i=0,\ldots,k-1$ and $0<d(x_k,x_{k+1})\le s.$

\begin{proposition} \label{monotone} If $G\in \CW(s,s-1)$ and $u,v$ are
two vertices of $G$ such that $\alpha(u)<\alpha(v)$ and $d(u,v)>
s^2,$ then for any shortest $(u,v)$-path $P(u,v),$ the vertex $x_1$
of its $s$-net $N(u,v)=(u=x_0,x_1,\ldots,x_k,x_{k+1}=v)$ satisfies
the condition $\alpha(u)<\alpha(x_1).$
\end{proposition}

\begin{proof} Suppose by way of contradiction that $\alpha(u)>\alpha(x_1).$ Let
$x_i$ $(1\le i\le k)$ be a vertex of $N(u,v)$ having a locally
minimal index $\alpha(x_i),$ i.e.,
$\alpha(x_{i-1})>\alpha(x_i)<\alpha(x_{i+1}).$ Let $y_i$ be the
vertex eliminating $x_i$ in the $(s,s-1)$-dominating order. We
assert that $d(y_i,x_{i-1})\le s-1$ and $d(y_i,x_{i+1})\le s-1.$
Indeed, if $y_i$ does not belong to the portion of the path $P(u,v)$
comprised between $x_{i-1}$ and $x_{i+1}$, then $x_{i-1},x_{i+1}\in
X_{\alpha(x_i)}\cap N_s(x_i, G\setminus \{ y_i\}),$ and therefore
$x_{i-1},x_{i+1}\in N_{s-1}(y_i)$ by the dismantling condition. Now
suppose that $y_i$ belongs to one of the segments of $P(u,v),$  say
to the subpath between $x_{i-1},x_i.$ Since $y_i\ne x_i$ we conclude
that $d(x_{i-1},y_i)\le s-1.$ On the other hand, since $x_{i+1}\in
X_{\alpha(x_i)}\cap N_s(x_i, G\setminus \{ y_i\}),$ by dismantling
condition we conclude that $d(y_i,x_{i+1})\le s-1.$ Hence, indeed
$d(y_i,x_{i-1})\le s-1, d(y_i,x_{i+1})\le s-1,$ whence
$d(x_{i-1},x_{i+1})\le 2s-2.$ Since $d(x_{i-1},x_{i+1})=2s$ for any
$1\le i\le k-1,$ we conclude that $i=k.$ Therefore the indices of
the vertices of $N(u,v)$  satisfy the inequalities
$\alpha(u)=\alpha(x_0)>\ldots>\alpha(x_{k-1})>\alpha(x_k)<
\alpha(x_{k+1})=\alpha(v).$

Denote by $N$ the ordered sequence of vertices
$x_0=u,x_1,\ldots,x_{k-1},y_k,x_{k+1}=v$ obtained from the $s$-net $N(u,v)$
by replacing the vertex $x_k$ by $y_k.$ We say that $N$ is obtained
from $N(u,v)$ by an {\it exchange}.  Call two consecutive vertices
of $N$ a {\it link}; $N$ has $k+1$ links, namely, $k-1$ links of length $s$
and two links of length at most $s-1.$ If
$\alpha(y_k)<\alpha(x_{k-1}),$ then we perform with $y_k$ the same
exchange operation as we did with $x_k.$ After several such exchanges,
we will obtain a new ordered set $x_0=u,x_1,\ldots,x_{k-1},z_k,x_{k+1}=v$ (denote
it also by $N$) having $k-1$ links of length $s$ and two links of
length $\le s-1$ and $\alpha(x_{k-1})<\alpha(z_k).$ Since
$\alpha(x_{k-2})>\alpha(x_{k-1}),$ using the $(s,s-1)$-dismantling
order we can exchange in $N$ the vertex $x_{k-1}$ by a vertex
$y_{k-1}$ to get an ordered set (denote it also by $N$) having $k-3$
links of length $s$ and 3 links of length $s-1$. Repeating the
exchange operation with each occurring local minimum (different from
$u$) of $N$ with respect to the total order $\alpha$,  after a
finite number of exchanges we will obtain an ordered set $N=(
u,z_1,z_2,\ldots,z_k,v)$  consisting of $k+1$ links of length at
most $s-1$ each and such that $\alpha(u)<\alpha(z_i)$ for any
$i=1,\ldots,k.$ By triangle inequality, $d(u,v)\le
d(u,z_1)+d(z_1,z_2)+\ldots+d(z_k,v)\le (k+1)(s-1).$ On the other
hand, from the definition of $N(u,v)$ we conclude that
$d(u,v)=ks+\gamma,$ where $0<\gamma=d(x_k,v)\le s.$ Hence
$(k+1)(s-1)\ge ks+\gamma,$ yielding $k\le s-\gamma-1.$ But then
$d(u,v)=ks+\gamma\le (s-\gamma-1)s+\gamma=s^2-s\gamma-s+\gamma<s^2,$
contrary to the assumption that $d(u,v)\ge s^2.$ This contradiction
shows that indeed $\alpha(x_1)>\alpha(u).$
\end{proof}

We call a graph $G\in \CW(s,s-1)$  $(s,s-1)^*$-{\it dismantlable} if for
any $(s,s-1)$-dismantling order $v_1,\ldots v_n$ of $G,$ for each  vertex
$v_i, 1\le i<n,$ there exists another vertex $v_j$
{\it adjacent} to $v_i$ such that $N_{s}(v_i,G\setminus\{ v_j\})\cap
X_i\subseteq N_{s-1}(v_j),$ where $X_i:=\{ v_i,v_{i+1},\ldots,v_n\}$
and $X_n=\{ v_n\}.$
The difference between $(s,s-1)$-dismantlability
and $(s,s-1)^*$-dismantlability is that in the second case the
vertex $v_j$ dominating $v_i$ is necessarily adjacent to $v_i$ but
not necessarily eliminated after $v_i$.


\begin{proposition} \label{hyp_dismantlable} If a graph $G\in \CW(s,s-1)$
is $(s,s-1)^*$-dismantlable, then $G$ is $s^2$-hyperbolic.
\end{proposition}

\begin{proof} Pick any quadruplet of vertices  $u,v,x,y$ of $G,$
consider its representation as in Fig.~\ref{fig0} where $\xi \le \eta,$
and proceed by induction on the total distance sum
$S(u,v,x,y)=d(u,v)+d(u,x)+d(u,y)+d(v,x)+d(v,y)+d(x,y).$ From Fig. \ref{fig0} we immediately conclude that if one of the distances
between the vertices $u,v,x,y$ is at most $s^2,$ then $\xi\le s^2$ and we are done. So suppose that the distance between any two
vertices of our quadruplet is at least $s^2.$

Consider any $(s,s-1)$-dismantling order $v_1,\ldots, v_n$ of $G$ and suppose that
$u$ is the vertex of our quadruplet occurring first in this order.
Pick three shortest paths $P(u,v),P(u,x),$ and $P(u,y)$ between the
vertex $u$ and the three other vertices of the quadruplet. Denote by
$v_1,x_1,$ and $y_1$ the vertices of the paths $P(u,v),P(u,x),$ and
$P(u,y),$  respectively, located at distance $s$ from $u$. From
Proposition \ref{monotone} we infer that $u$ is eliminated before
each of the vertices $v_1,x_1,y_1.$ Let $u'$ be the neighbor of $u$
eliminating $u$ in the $(s,s-1)^*$-dismantling order associated with the $(s,s-1)$-dismantling order
$v_1,\ldots,v_n.$ From the $(s,s-1)^*$-dismantling condition we infer
that each of the distances $d(u',v_1),d(u',x_1),d(u',y_1)$ is at
most $s-1.$ Since $u$ is adjacent to $u'$ and $u$ is at distance $s$
from $v_1,x_1,y_1,$ necessarily $d(u',v_1),d(u',x_1),d(u',y_1)$ are
all equal to $s-1.$  Therefore, if we will replace in our quadruplet
the vertex $u$ by $u',$ we will obtain a quadruplet with a smaller total
distance sum: $S(u',v,x,y)=S(u,v,x,y)-3.$ Therefore, by
induction hypothesis, the two largest of the distance sums
$d(u',v)+d(x,y)$, $d(u',x)+d(v,y)$, $d(u',y)+d(v,x)$ differ by at
most $2s^2$. On the other hand, $d(u,v)+d(x,y)=d(u',v)+d(x,y)+1$,
$d(u,x)+d(v,y)=d(u',x)+d(v,y)+1$, and
$d(u,y)+d(v,x)=d(u',y)+d(v,x)+1,$ whence the two largest distance
sums of the quadruplet $u,v,x,y$ also differ by at most $2s^2.$
Hence $G$ is $s^2$-hyperbolic.
\end{proof}

A graph $G$ is called a {\it Helly graph} if its family of balls
satisfies the Helly property: any collection of pairwise
intersecting balls has a common vertex. A graph $G$ is called a {\it
bridged graph} if all isometric cycles of $G$ have length 3.
Equivalently, $G$ is a bridged graph if all balls around convex sets
are convex (a subset $S$ of vertices is convex if together with any
two vertices $u,v,$ the set $S$ contains the {\it interval}
$I(u,v)=\{ x\in V: d(u,v)=d(u,x)+d(x,v)\}$ between $u$ and $v$). For
a comprehensive survey of results and bibliography on Helly and
bridged graphs, see \cite{BaCh_survey}.

\begin{figure}
\begin{center}
\scalebox{0.7}
{\input{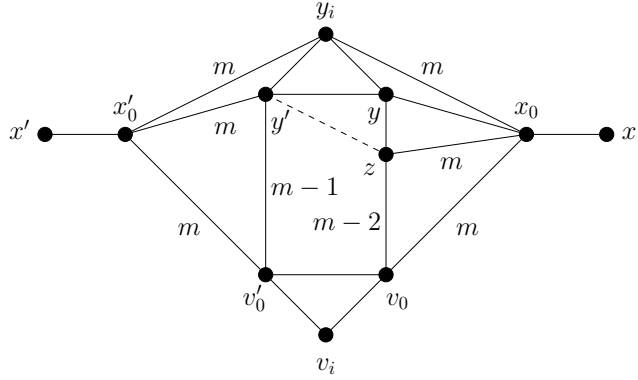}}
\end{center}
\caption{To the proof of Proposition~\ref{HellyBridged} (case of bridged graphs).}
\label{figBridged}
\end{figure}

\begin{proposition}\label{HellyBridged} If $G\in \CW(s,s-1)$ is a Helly or a bridged graph, then $G$ is
$(s,s-1)^*$-dismantlable and therefore $G$ is $s^2$-hyperbolic.
\end{proposition}

\begin{proof}
The second assertion immediately follows from Proposition~\ref{hyp_dismantlable}.
Thus, we only need to prove that any Helly or bridged graph in $\CW(s,s-1)$ is $(s,s-1)^*$-dismantlable.

First, let $G$ be an $(s,s-1)$-dismantlable Helly graph. Let
$v_i$ be the $i$th vertex in an $(s,s-1)$-dismantling order and let
$y_i$ be the vertex eliminating $v_i.$ Suppose that
$k:=d(v_i,y_i)\ge 2.$ We assert that we can always eliminate $v_i$
with a vertex $y'_i$ adjacent to $y_i$ and located at distance $k-1$
from $v_i.$ Then repeating the same  reasoning with $y'_i$ instead
of $y_i,$ we will eventually arrive at a vertex of $I(v_i,y_i)$
adjacent to $v_i$ which still eliminates $v_i.$ Set $A:=(X_i\cap
N_s(v_i))\setminus \{v_i,y_i\}.$ For each vertex $x\in A,$ consider
the ball $N_{s-1}(x)$ of radius $s-1$ centered at $x.$ Consider also
the balls $N_{k-1}(v_i)$ and $N_1(y_i).$ We assert that the balls of
the resulting collection pairwise intersect. Indeed, any two balls
centered at vertices of $A$ intersect in $y_i.$ The ball $N_1(y_i)$
intersects  any ball centered at $A$ in $y_i$. The ball
$N_{k-1}(v_i)$ intersects any ball centered at a vertex $x\in A$
because $d(v_i,x)\le s\le k-1+s-1.$ Finally, $N_{k-1}(v_i)$ and
$N_1(y_i)$ intersect because  $d(v_i,y_i)=k=k-1+1.$ By Helly
property, the balls of this collection intersect in a vertex $y'_i.$
Since $y'_i$ is at distance at most $k-1$ from $v_i$ and at distance
at most 1 from $y_i,$ from the equality $d(v_i,y_i)=k$ we
immediately deduce that $y'_i$ is a neighbor of $y_i$ located at
distance $k-1$ from $v_i.$ This establishes the
$(s,s-1)^*$-dismantling property for Helly graphs in $\CW(s,s-1)$.

Now, suppose that $G$ is a bridged graph and let the vertices
$v_i,y_i$ and the set $A$ be defined as in the previous case. Since
$G$ is bridged, the convexity of the ball $N_{k-1}(v_i)$ implies that
the set $C$ of neighbors of $y_i$ in the interval $I(v_i,y_i)$ induces
a complete subgraph. Pick any vertex $x\in A.$ Clearly, $d(x,y_i)\le
s-1$ and $d(x,v_i)\le s.$ If $d(x,v_i)\le s-1,$ then $v_i,y_i\in
N_{s-1}(x)$ and from the convexity of the ball $N_{s-1}(x)$ we
conclude that $I(v_i,y_i)\subset N_{s-1}(x).$ Hence, in this case,
$d(x,y)\le s-1$ for any $y\in I(v_i,y_i),$ in particular, for any
vertex of $C.$ Analogously, if $d(x,y_i)<s-1,$ then $d(x,y)\le s-1$
for any vertex $y\in C.$ Therefore the choice of the vertex $y'_i$ in
$C$ depends only of the vertices of the set $A_0=\{ x\in A: d(x,v_i)=s
\mbox{ and } d(x,y_i)=s-1 \}.$

Pick any
vertex $x\in A_0.$  If $I(x,y_i)\cap I(y_i,v_i)\ne \{ y_i\},$ then
$y_i$ has a neighbor $y'$ in this intersection located at distance
$s-2$ from $x.$ Since $y'\in C$ and $C$ is a complete subgraph, then
$d(y,x)\le s-1$ for any $y\in C.$ Therefore we can discard all such
vertices of $A_0$ from our future analysis and suppose
without loss of generality that $I(x,y_i)\cap I(y_i,v_i)=\{ y_i\}$
for any $x\in A_0.$ For $x\in A_0,$ let $x_0$ be a furthest from $x$
vertex of $I(x,y_i)\cap I(x,v_i).$ Let $v_0$ be a furthest from
$v_i$ vertex of $I(v_i,x_0)\cap I(v_i,y_i).$ Since $I(x,y_i)\cap
I(y_i,v_i)=\{ y_i\}$ and $G$ is bridged, the vertices $y_i,x_0,v_0$
define an equilateral metric triangle sensu \cite{BaCh_weak,BaCh_survey}:
$d(y_i,x_0)=d(x_0,v_0)=d(v_0,y_i)=:m.$ Moreover, any vertex of
$I(v_0,y_i)$ is located at distance $m$ from $x_0$  and therefore
at distance  $s-1$ from $x,$ showing, in particular, that
$N_{s-1}(x)\cap C\ne\emptyset$ for any $x\in A_0$.
From the definition of $x_0$ and $v_0$
we conclude that $m+d(x_0,x)=s-1,
d(x,x_0)+m+d(v_0,v_i)=s,$ and $d(v_i,v_0)+m\le s-1.$ Whence
$d(v_i,v_0)=1,$ yielding $d(v_i,y_i)=m+1.$

Pick in $C$ a vertex $y$ belonging to a maximum number of balls
$N_{s-1}(x)$ centered at $x\in A_0.$ Suppose by way of contradiction
that $A_0$ contains a vertex $x'$ such that $y\notin N_{s-1}(x')$ (for an illustration, see
Fig.~\ref{figBridged}).
Since $d(x',y_i)=s-1$ and $y$ is adjacent to $y_i,$ we have
$d(x',y)=s.$ Let $y'$ be a vertex of $C$ belonging to $N_{s-1}(x')$
(such a vertex $y'$ exists because of the remark in above paragraph).
Let $v'_0$ be the neighbor of $v_i$ defined with respect to $x'$ in
the same way as $v_0$ was defined for $x$. Then all vertices of
$I(v'_0,y')$ are located at distance $s-1$ from $x'.$ We can suppose
that there exists a vertex $x\in A_0$ such that $y\in N_{s-1}(x)$
but $y'\notin N_{s-1}(x),$ otherwise we will obtain a contradiction
with the choice of $y.$ Since the balls $N_{s-1}(x)$ and
$N_{s-1}(x')$ are convex, the intervals $I(v_0,y_i)$ and
$I(v'_0,y_i)$ belong to these balls, respectively, whence
$d(v_0,y)=d(v'_0,y')=m-1$ but $d(v_0,y')=d(v_0',y)=m.$ Let $z$
be a neighbor of $y$ in $I(v_0,y).$
Since $z,y'\in I(y,v'_0)$ and $G$ is bridged, the vertices $z$ and
$y'$ are adjacent. Hence $y'\in I(v_0,y_i),$ yielding $d(x,y')=s-1,$
contrary to our assumption that $y'\notin N_{s-1}(x).$
This contradiction shows that $C$
contains a vertex belonging to all balls $N_{s-1}(x)$ centered at
vertices of $A_0,$ thus establishing the $(s,s-1)^*$-dismantling
property for bridged graphs in $\CW(s,s-1)$.
\end{proof}


\begin{proposition} \label{hyps>2s'} If $s\ge 2s',$ then any graph $G$ of $\CW(s,s')$ is
$(s-1)$-hyperbolic.
\end{proposition}
\begin{proof}  First, similarly to Proposition \ref{monotone}, we prove
that if $d(u,v)\ge s$ and $\alpha(u)<\alpha(v),$ then the vertex $x_1$ of
the $s$-net $N(u,v)$ of any shortest $(u,v)$-path satisfies the
inequality $\alpha(x_1)>\alpha(u).$ Suppose by way of contradiction
that $\alpha(u)>\alpha(x_1).$ Then as in proof of Proposition
\ref{monotone} we conclude that $x_k$ is the unique local minimum of
$\alpha$ on $N(u,v):$ $\alpha(x_{k-1})>\alpha(x_k)<\alpha(x_{k+1}).$
Let $y_k$ be the vertex eliminating $x_k$ in the $(s,s')$-dominating
order.  If $y_k$ does not belong to the segment of $P(u,v)$  between $x_{k-1}$ and $x_k,$
then $d(x_{k-1},x_{k+1})\le d(x_{k-1},y_k)+d(y_k,x_{k+1})\le 2s',$
contrary to the assumption that $d(x_{k-1},x_{k+1})>s\ge 2s'.$ So
$y_k$ belongs to the subpath of $P(u,v)$ between $x_{k-1}$ and
$x_{k+1}.$ If $y_k$ belongs to the subpath comprised between $x_k$
and $x_{k+1},$ then the dismantling condition implies that
$d(y_k,x_{k-1})\le s',$ which is impossible because
$d(y_k,x_{k-1})=d(y_k,x_k)+s>2s'.$ The same contradiction is
obtained if $y_k$ belongs to the second half of the subpath
between $x_{k-1}$ and $x_k.$ Finally, if $y_k$ belongs to
the first half of this subpath, then $d(y_k,x_{k+1})\le s'$ by the
dismantling condition, contradicting the fact that the location of
$y_k$ on this subpath of $P(u,v)$ implies that $d(y_k,x_{k+1})>s'.$ This shows that
indeed $\alpha(x_1)>\alpha(u).$

To establish $(s-1)$-hyperbolicity of $G,$ as in the proof of
Proposition \ref{hyp_dismantlable} we pick any quadruplet of
vertices $u,v,x,y$ of $G$  and proceed by induction on the total
distance sum $S(u,v,x,y)=d(u,v)+d(u,x)+d(u,y)+d(v,x)+d(v,y)+d(x,y).$
Again, we can suppose that  the distances between any two vertices
of this quadruplet is at least $s,$ otherwise we are done. Let $u$ be
the vertex of our quadruplet occurring first in some
$(s,s')$-dismantling order of $G$. Pick three shortest paths
$P(u,v),P(u,x),$ and $P(u,y)$ and denote by $v_1,x_1,$ and $y_1$ their
respective vertices located at distance $s$ from $u.$ From first
part of our proof we infer that $u$ is eliminated before
$v_1,x_1,$ and $y_1.$ Let $u'$ be the vertex  eliminating $u.$ From the
$(s,s')$-dismantling condition we infer that $d(u,u')\le s'.$
Moreover, either $d(u',v_1)\le s'$ or $v_1 \notin N_s(u,G\setminus \{
u'\}).$ Since $d(u,v_1)=s\ge 2s',$ in both cases we conclude that
$u'$ belongs to a shortest $(u,v_1)$-path of $G$. Analogously, we
conclude that $u'$ lie on a shortest $(u,x_1)$-path and on a shortest
$(u,y_1)$-path. Therefore, if we replace in our quadruplet $u$
by $u',$ we will get a quadruplet with total distance sum
$S(u',v,x,y)=S(u,v,x,y)-3d(u,u')<S(u,v,x,y).$ By induction hypothesis,
the two largest distance sums of this quadruplet differ by at most $2(s-1).$
On the other hand, since $d(u,v)+d(x,y)=d(u',v)+d(x,y)+d(u,u')$,
$d(u,x)+d(v,y)=d(u',x)+d(v,y)+d(u,u')$, and
$d(u,y)+d(v,x)=d(u',y)+d(v,x)+d(u,u'),$ the two largest distance
sums of the quadruplet $u,v,x,y$ also differ by at most $2(s-1).$
Hence $G$ is $(s-1)$-hyperbolic.
\end{proof}

\section{Cop-win graphs for game with fast robber: class $\CW(s)$}

In this section, we specify the dismantling scheme provided by
Theorem~\ref{copwin} in order to characterize the graphs in which one
cop with speed $1$ captures a robber with speed $s\ge 2$.  First we
show that the graphs from $\CW(2)$ are precisely the dually chordal
graphs \cite{BraDraCheVol}.  Then we show that for  $s\ge 3$ the
classes $\CW(s)$ coincide with $\CW(\infty)$ and we provide a
structural characterization of these graphs.

\subsection{$\CW(2)$ and dually chordal graphs}

We start by showing that when the cop has speed $1$ and the robber has speed $s\ge 1$,
then the dismantling order in Theorem \ref{copwin} can be defined using the subgraphs $G_i=G(X_i).$

\begin{proposition}\label{prop-sans-X}
A graph $G$ is $(s,1)$-dismantlable if and only if the vertices of $G$
can be ordered $v_1, \ldots, v_n$ in such a way that for each vertex $v_i\ne v_n$
there exists a vertex $v_j$ with $j>i$ such that $N_s(v_i,
G_i\setminus\{v_j\})) \subseteq N_1(v_j,G_i).$
\end{proposition}

\begin{proof}
First, note that for any $i \leq j$, $N_1(v_j,G) \cap X_i =
N_1(v_j,G_i)$. Thus, if a graph $G$ is $(s,1)$-dismantlable, then any
$(s,1)$-dismantling order satisfies the requirement $N_s(v_i,
G_i\setminus\{v_j\})) \subseteq N_1(v_j,G_i).$
Conversely, consider an order $v_1, \ldots, v_n$ on the vertices of
$G$ satisfying this condition.    If $s = 1$, then $N_1(v_i, G_i\setminus\{v_j\})) =
N_1(v_i,G\setminus\{v_j\}) \cap X_i$ and thus our assertion is
obviously true. We now suppose that $s \geq 2$. By induction on $i,$ we will show that $N_s(v_i,G\setminus\{v_j\}) \cap X_i
\subseteq N_1(v_j)$. For $i = 1$, $G_i = G$ and thus the property
holds. Consider $i$ such that for any $i' < i$, the property is
satisfied. Pick any vertex $u \in N_s(v_i) \cap X_i$.  If the distance in $G_i\setminus\{v_j\}$
between $v_i$ and $u$ is at most $s,$ then $u \in
N_s(v_i,G_i\setminus\{v_j\}) \subseteq N_1(v_j)$ and we are
done. Otherwise, we can find a unique index $i_0 < i$ such that the distance between $v_i$
and $u$ in the graph $G_{i_0}\setminus\{v_j\}$ is at most $s$ and in the graph
$G_{i_0+1}\setminus\{v_j\}$ is larger than $s.$  Consider a shortest path
$\pi$ between $v_i$ and $u$ in
$G_{i_0}\setminus\{v_j\}$. From the choice of $i_0$, necessarily $v_{i_0}$ is a vertex
of $\pi$. Since the length of $\pi$ is at most $s,$ we deduce that
$d_{G_{i_0}}(u,v_{i_0}) \leq s$ and $d_{G_{i_0}}(v_i,v_{i_0}) \leq
s$. By the induction hypothesis, there exists $j_0 > i_0$ such
that $N_s(v_{i_0}, G_{i_0}\setminus\{v_{j_0}\})) \subseteq
N_1(v_{j_0})$. If $j_0 \neq j$, then there exists a path
$(u,v_{j_0},v_i)$ of length $2$ between $u$ and $v_i$ in $G_{j_0}.$ Since $j_0>i_0,$
we obtain a contradiction with the definition of $i_0.$
Hence $j_0 = j$, and, by our induction
hypothesis, $u \in N_s(v_{i_0}, G_{i_0}\setminus\{v_{j}\})) \subseteq
N_1(v_{j})$, and we are done.
\end{proof}

Analogously to Theorem 3 of Clarke \cite{Clarke} for the witness
version of the game, it can be easily shown that, for any $s$, the
class $\CW(s)$ is closed under retracts:

\begin{proposition}\label{prop-retract}
If $G \in \CW(s)$ and $G'$ is a retract of $G$, then $G'
\in \CW(s)$.
\end{proposition}

Recall that a graph $G$
is called {\it dually chordal} \cite{BraDraCheVol} if its clique
hypergraph (or, equivalently, its ball hypergraph) is a hypertree,
i.e., it satisfies the Helly property and its line graph is
chordal (see the Berge's book on hypergraphs \cite{Be} for these two
definitions). Dually chordal graphs are equivalently defined as the graphs
$G$ having a spanning tree $T$ such that any maximal clique or any
ball of $G$ induces a subtree of $T.$ Finally, dually chordal graphs
are exactly the graphs $G=(V,E)$ admitting a maximum neighborhood
ordering of its vertices. A vertex $u\in N_1(v)$ is a {\it maximum
  neighbor} of $v$ if for all $w\in N_1(v)$ the inclusion
$N_1(w)\subseteq N_1(u)$ holds. The ordering $\{ v_1,\ldots,v_n\}$ is a
{\it maximum neighborhood ordering} (\emph{mno} for short) of $G$
\cite{BraDraCheVol}, if for all $i<n,$ the vertex $v_i$ has a maximum
neighbor in the subgraph $G_i$ induced by the vertices $X_i=\{
v_i,v_{i+1},\ldots,v_n\}$. Dually chordal graphs comprise strongly
chordal graphs, doubly chordal, and interval graphs as subclasses and can be
recognized in linear time. Any
graph $H$ can be transformed into a dually chordal graph by adding a
new vertex $c$ adjacent to all vertices of $H.$


\begin{theorem} \label{dually_chordal} For a graph $G=(V,E),$ the following conditions are equivalent:
\begin{itemize}
\item[(i)] $G\in \CW(2)$;
\item[(ii)] $G$ is $(2,1)$-dismantlable;
\item[(iii)] $G$ admits an mno ordering;
\item[(iv)] $G$ is dually chordal.
\end{itemize}
\end{theorem}

\begin{proof} Since $\CW(2)=\CW(2,1),$ the equivalence (i)$\Leftrightarrow$(ii) follows from Theorem \ref{copwin}. The equivalence
(iii)$\Leftrightarrow$(iv) is a result of \cite{BraDraCheVol}. Notice
  that $u$ is a maximum neighbor of $v$ in $G$ iff $N_2(v)=N_1(u).$
  Therefore, $\{ v_1,\ldots,v_n\}$ is a maximum neighborhood ordering
  of $G$ iff for all $i<n,$ $N_2(v_i,G_i)=N_1(v_j, G_i)$ for some
  $v_j, j>i$. Hence any mno ordering is a $(2,1)$-dismantling
  ordering, establishing (iii)$\Rightarrow$(ii).  Finally, by
  induction on the number of vertices of $G$ we will show that any
  $(2,1)$-dismantling ordering $\{ v_1,\ldots,v_n\}$ of the vertex set
  of $G$ is an mno, thus (ii)$\Rightarrow$(iii).  Suppose that
  $N_2(v_1,G\setminus\{ u\})\subset N_1(u)$ for some $u:=v_j, j>1.$
  Then $u$ is adjacent to $v_1$ and to all neighbors of $v_1.$ Since
  for any neighbor $w\ne u$ of $v_1$ the ball $N_1(w)$ is contained in
  the punctured ball $N_2(v_1,G\setminus\{ u\}),$ we conclude that
  $N_1(w)\subseteq N_1(u),$ i.e., $u$ is a maximum neighbor of $v_1.$
  The graph $G'$ obtained from $G$ by removing the vertex
  $v_1$ is a retract, and therefore an isometric subgraph of $G$.
  Thus for any vertex $v_i, i > 1$, by what has been noticed above (Proposition~\ref{prop-sans-X}),
  the intersection of a ball (or of a punctured ball) of $G$ centered
  at $v_i$ with the set $X_2 = \{ v_2,\ldots,v_n\}$ coincides with the
  corresponding ball (or punctured ball) of the graph $G' = G(X_2)$
  centered at the same vertex $v_i$.  Therefore $\{ v_2,\ldots,v_n\}$
  is a $(2,1)$-dismantling ordering of the graph $G'$. By induction
  assumption, $\{ v_2,\ldots,v_n\}$ is an mno of $G'$. Since $v_1$ has
  a maximum neighbor in $\{ v_2,\ldots,v_n\},$ we conclude that $\{
  v_1,v_2,\ldots v_n\}$ is a maximum neighborhood ordering of $G$.
\end{proof}

\subsection{$\CW(k), k\ge 3,$ and big brother graphs}

A {\it block} of a graph $G$ is a maximal by inclusion vertex
two-connected subgraph of $G$ (possibly reduced to a single edge). Two
blocks of $G$ are either disjoint or share a single vertex, called an
{\it articulation point}.  Any graph $G=(V,E)$ admits a
block-decomposition in the form of a rooted tree $T$: each vertex of
$T$ is a block of $G,$ pick any block $B_1$ as a root of $T,$ label
it, and make it adjacent in $T$ to all blocks intersecting it, then
label that blocks and make them adjacent to all nonlabeled blocks
which intersect them, etc. A block $B$ of $G$ is {\it dominated} if it
contains a vertex $u$ (called the {\it big brother} of $B$) which is
adjacent to all vertices of $B.$ A graph $G$ is a {\it big brother
graph}, if its block-decomposition can be represented in the form of a
rooted tree $T$ is such a way that (1) each block of $G$ is dominated
and (2) for each block $B$ distinct from the root $B_1$, the
articulation point between $B$ and its father-block dominates $B.$
Equivalently, $G$ is a big brother graph if its blocks can be ordered
$B_1,\ldots,B_r$ such that $B_1$ is dominated and, for any $i>1$, the
block $B_i$ is a leaf in the block-decomposition of $\cup_{j \leq i}
B_j$ and is dominated by the articulation point connecting $B_i$ to
$\cup_{j < i} B_j$ (we will call such a decomposition a
\emph{bb-decomposition} of $G$); see Fig.~\ref{bigBrother}(a) for
an example.

\begin{figure}[t]
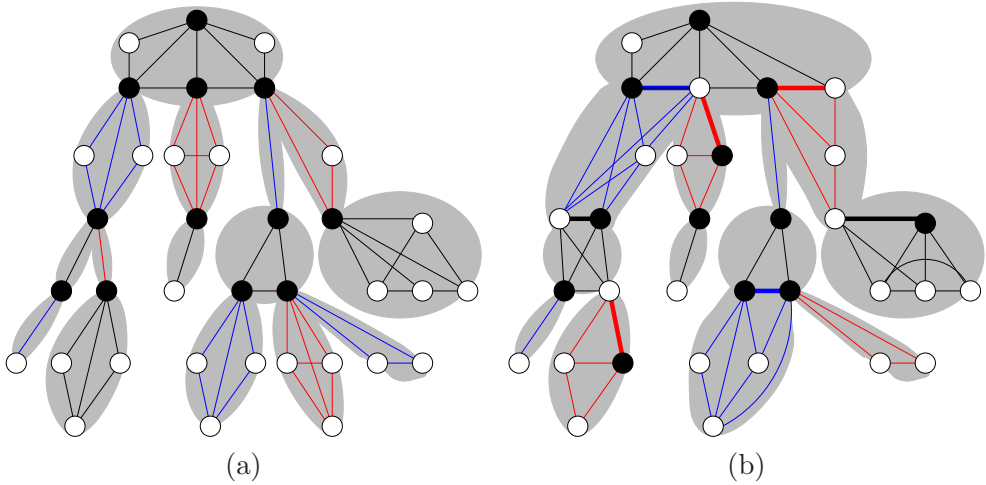

\begin{center}
\begin{tabular}{cc}
{\input{bigBrother.pstex_t}}&
{\input{cww.pstex_t}}\\
(a) & (b)
\end{tabular}
\end{center}
\caption{(a) A big brother graph. (b) A big two-brother graph.}
\label{bigBrother}
\end{figure}

\begin{theorem} \label{th-speed-3}
For a graph $G=(V,E)$ the following conditions are equivalent:
\begin{itemize}
\item[(i)] $G \in \CW(3)$; \item[(i$'$)] $G$ is $(3,1)$-dismantlable;
\item[(ii)] $G \in \CW(\infty);$ \item[(ii$'$)] $G$ is $(\infty,1)$-dismantlable;
\item[(iii)] $G$ is a big brother graph.
\end{itemize}
In particular, the classes of graphs $\CW(s), s\ge 3,$ coincide.
\end{theorem}

\begin{proof} The equivalences (i)$\Leftrightarrow$(i$'$) and (ii)$\Leftrightarrow$(ii$'$) are particular cases of Theorem \ref{copwin}.
Next we will establish (iii)$\Rightarrow$(i)$\&$(ii), i.e., that any
big brother graph $G$ belongs to $\CW(s)$ for all $s\ge 3$. Let
$B_1,\ldots,B_r$ be a bb-decomposition of $G$. We consider the
following strategy for the cop. At the beginning of the game, we
locate the cop at the big brother of the root-block $B_1$. Now, at
each subsequent step, the cop moves to the neighbor of his current
position that is closest to the position of the robber. Notice the
following invariant of the strategy: the position of the cop will
always be at the articulation point of a block $B$ on the path of $T$
between the previous block hosting $\mathcal C$ and the current block
hosting $\mathcal R$.  This means that, since $\mathcal R$ cannot
traverse this articulation point without being captured, $\mathcal R$
is restricted to move only in the union of blocks in the subtree
rooted at $B.$ Now, if before the move of  the cop, $\mathcal C$ and
$\mathcal R$ occupy their positions in the same block, then
$\mathcal C$ captures $\mathcal R$ at the next move. Otherwise, the
next move will increase the distance in $T$ between the root and the
block hosting $\mathcal C$. Therefore after at most diameter of $T$
rounds, $\mathcal R$ and $\mathcal C$ will be located in the same
block, and thus the cop captures the robber at next move. This shows
that (iii)$\Rightarrow$(i)$\&$(ii).

The remaining part of the proof is devoted to the implication
(i)$\&$(i$'$)$\Rightarrow$(iii). Let $G$ be a graph of $\CW(3).$
Notice first that for any articulation point $u$ of $G$, and any
connected component $C$ of $G \setminus \{u\}$, the graph induced by
$C \cup \{u\}$ also belongs to $\CW(3)$. Indeed, this follows by
noticing that $G(C \cup \{u\})$ is a retract of $G$ (this retraction
is obtained by mapping all vertices outside $C$ to $u$) and that
$\CW(3)$ is closed under retracts by
Proposition~\ref{prop-retract}.  To prove that a graph $G=(V,E)\in
\CW(3)$ is a big brother graph, we will proceed by induction on the
number of vertices of $G.$ If $G$ has one or two vertices, the result
is obviously true. For the inductive step, we distinguish two cases,
depending if $G$ is two-connected or not.

\vspace{1ex}
\noindent{\bf Case 1:}  $G$ is not two-connected.

Since each block of $G$ has strictly less vertices than $G$, by
induction hypothesis each block is a big brother graph, i.e.,
it has a dominating vertex. First suppose that the
block-decomposition of $G$ has a leaf $B$ such that the articulation
point $a$ of $B$ separating $B$ from the rest of $G$ is a big brother
of $B.$ Let $G'$ be the subgraph of $G$ induced by all blocks of $G$
except $B,$ i.e., $G'=G(V\setminus (B\setminus\{a\})).$ Since $G'\in
\CW(3)$ by what has been shown above, from the induction hypothesis we
infer that $G'$ is a big brother graph. Consequently, there exists a
bb-decomposition $B_1, \ldots, B_r$ of $G'$. Then, $B_1, \ldots, B_r,
B$ is a bb-decomposition of $G$ and thus, $G$ is a big brother
graph. Suppose now that for any leaf in the block-decomposition of
$G$, the articulation point of the corresponding block does not
dominate it. Pick two leaves $B_1$ and $B_2$ in the
block-decomposition of $G$ and consider their unique articulation
points $a_1$ and $a_2$ ($a_i$ disconnects $B_i$ from the rest of
$G$). We claim that in this case, a robber that moves at speed $3$ can
always escape, which will contradicts the assumption that $G \in
\CW(3)$. Let $b_i$ be the dominating vertex of the block $B_i$, $i =
1, 2$ (by assumption, $b_i \neq a_i$). Consider now a vertex $c_i\in
B_i \setminus \{b_i\}$ which can be connected with $a_i$ by a 2-path
$(c_i,g_i,a_i)$ avoiding $b_i$ (such a vertex exists because $B_i$ is
two-connected and, by assumption, $a_i$ is not a dominating vertex of
$B_i$). Let $\pi$ be a shortest path from $a_1$ to $a_2$ in $G$ and
let $h_1$ and $h_2$ be the neighbors in $\pi$ of $a_1$ and $a_2,$
respectively.  Note that $h_i$ does not belong to $B_i,$ thus $a_i$ is
the only neighbor of $h_i$ in $B_i$. We now describe a strategy that enables
the robber to escape. Initially, if the cop is not in $B_1$, then the robber
starts in $c_1$; otherwise, he starts in $c_2$. Then the robber stays
in $c_i$, as long as the cop is at distance $\ge 2$ from
$c_i$. When the cop moves to a neighboring vertex of $c_i$, then the
robber goes to $h_i$ (either via the path $(c_i,b_i,a_i,h_i)$ or via
the path $(c_i,g_i,a_i,h_i)$) and then,  no matter how the cop moves,
he goes to $c_{3-i}$ using the
shortest path $\pi$. Now notice that when
$\mathcal R$ is in $h_i$, $\mathcal C$ is in $B_i \setminus \{a_i\}$
and thus he cannot capture the robber. When the robber is moving from
$h_i$ to $c_{3-i}$, he uses a shortest path $\pi$ of $G$: the cop
cannot capture him either because he is initially at distance $2$ from
the robber and he moves slower than the robber. Consequently, the cop
cannot capture the robber, contrary with the assumption $G \in
\CW(3)$.

\vspace{1ex}
\noindent{\bf Case 2:}  $G$ is  two-connected.

We must show that $G$ has a dominating vertex. Consider a
$(3,1)$-dismantling order $v_1, \ldots, v_n$ of the vertices of $G.$
Let $u$ be a vertex such that $N_3(v_1,G\setminus\{u\}) \subseteq
N_1(u).$ Since $u$ is a maximum neighbor of $v_1,$ the isometric
subgraph $G':= G(V\setminus\{v_1\})$ of $G$ also belongs to $\CW(3)$
because $v_2,\ldots, v_n$ is a $(3,1)$-dismantling ordering of $G'.$
By induction hypothesis, $G'$ is a big brother graph.  Again, we
distinguish two subcases, depending on the two-connectivity of
$G'$. First suppose that $G'$ is two-connected. Since $G'$ is a big
brother graph, it contains a dominating vertex $t.$ If $t$ is adjacent
to $v_1$, then $t$ dominates $G$ and we are done. Otherwise, consider
a neighbor $w \neq u$ of $v_1$. Any vertex $x\ne u$ of $G$ can be
connected to $v_1$ by the path $(v_1,w,t,x)$ of length 3 avoiding $u,$
thus $x$ belongs to the punctured ball $N_3(v_1,G \setminus \{u\}).$
As a consequence, $x$ is a neighbor of $u,$ thus $u$ dominates
$G$. Now suppose that $G'$ is not two-connected. We assert that $u$ is
the only articulation point of $G'.$ Assume by way of contradiction
that $w\ne u$ is an articulation point of $G'$ and let $x$ and $y$ be
two vertices of $G'$ such that all paths connecting $x$ to $y$ go
through $w.$ In $G$, $x$ and $y$ can be connected by two
vertex-disjoint paths $\pi_1$ and $\pi_2.$ Assume without loss of
generality that $w\notin \pi_1.$ Since $\pi_1$ cannot be a path of
$G',$ the vertex $v_1$ belongs to $\pi_1$. Let $\pi_1 = (x,
x_1,\ldots, x_k,v_1,y_l, \ldots, y_1)$. Since $x_k,y_l\in N_1(v_1)
\subseteq N_3(v_1,G\setminus\{u\}) \cup \{u\}\subseteq N_1(u)$,
necessarily $x_k, y_l\in N_1(u)$. If $x_k = u$ or $y_l=u$, then $(x,
x_1,\ldots, x_k,y_l,\ldots,y_1)$ is a path between $x$ and $y$ in
$G'\setminus\{w\},$ which is impossible. Thus $u$ is different from
$x_k$ and $y_l$ but adjacent to these vertices. But then $(x,
x_1,\ldots, x_k,u,y_l, \ldots, y_1)$ is a path from $x$ to $y$ in
$G'\setminus\{w\}$, leading again to a contradiction.  This shows that
$w$ cannot be an articulation point of $G'.$ Since $G'$ is not
two-connected, we conclude that $u$ is the only articulation point of
$G'$. By the induction hypothesis, any block $B$ of $G'$ is dominated
by some vertex $b.$ Suppose that $u$ does not dominate $G',$ for
instance, $u$ is not adjacent to some vertex $t$ of $B.$ Since $u$ is
the unique articulation point of $G'$ but is not an articulation point
of $G$, $v_1$ necessarily has a neighbor $w\neq u$ in $B$. Hence,
there is a path $(v_1,w,b,t)$ of length $3$ in $G\setminus\{u\}$ and
thus $t$ is a neighbor of $u$, because $t\in N_3(v_1,G\setminus\{u\})
\subseteq N_1(u)$.  Thus $u$ dominates $G'=G \setminus \{v_1\}$, and, since $v_1 \in N_1(u),$
$u$ dominates $G$ as well. This concludes the
analysis of Case 2 and the proof of the theorem. \end{proof}

\section{Cop-win graphs for game with witness: class $\bigcap_{k\ge 1} {\mathcal C}{\mathcal
  W}{\mathcal W}(k)$}\label{sec-witnessI}

In this and next sections, we investigate the structure of
$k$-winnable graphs. In analogy with big brother graphs, we
characterize here the graphs $G$ that are $k$-winnable for all $k \geq
1$, i.e., the graphs from the intersection $\bigcap_{k\ge 1} {\mathcal
  C}{\mathcal W}{\mathcal W}(k)$.

\subsection{Game with witness: preliminaries} In the $k$-witness version of the game, the cop first selects his
initial position and then the robber selects his initial position
which is visible to the cop. As in the classical cop and robber game, the
players move alternatively along an edge or pass. However, the robber
is visible to the cop only every $k$ moves. After having seen the
robber, the cop decides a sequence of his next $k$ moves (the first
move of such a sequence is called a \emph{visible} move).  The cop
captures the robber if they both occupy the same vertex at the same step
(even if the robber is invisible). In particular, the cop can capture
the visible robber if after the robber shows up, they occupy two
adjacent vertices of the graph.  Since we are looking for winning
strategies for the cop, we may assume that the robber knows the cop's
strategy, i.e., after each visible move, the robber knows the next
$k-1$ moves of the cop.  In the $k$-witness version of the game, a
\emph{strategy} for the cop is a function $\sigma$ which takes as an
input the $i$ first visible positions of the robber and the $ik$ first
moves of the cop and outputs the next $k$ moves of the cop. A winning
strategy is defined as before and in any $k$-winnable graph, the cop
has a positional winning strategy.  We will call a \emph{phase} of the
game the movements of the two players comprised between two
consecutive visible moves.  We will call the behavior of the cop
during several consecutive moves of the same phase $\{a,b\}$-{\it
  oscillating} if his moves alternate between the adjacent vertices
$a$ and $b.$ In a $k$-winnable graph $G$, given a winning cop's
strategy $\sigma$, any trajectory $S_r$ of the robber ends up in a
vertex $r_p$ at which the robber is captured.  We will say that the
trajectory $S_r = (r_1, \ldots, r_p)$ is {\it maximal} if $(r_1,
\ldots, r_{p-1})$ cannot be extended to a longer trajectory for
which the robber is not captured by the cop. Notice that the last
vertex $r_p$ in a maximal trajectory $S_r$ corresponds to an
invisible move if and only if it is a leaf of $G$. Indeed, otherwise
let $r_{p-1}$ be the previous position of the robber. If $r_{p-1}
\neq r_p$, the robber could have stayed in $r_{p-1}$ to avoid being
captured. Thus $r_{p-1} = r_p$ and if $r_p$ has at least two
neighbors, the robber can safely move to one of the neighbors of
$r_p$ not occupied by the cop, and survive for an extra unit of
time.  We continue with two simple observations, the first shows
that during a phase an invisible robber can always safely move
around a cycle, while the second shows that  a robber visiting one
of the vertices $a$ or $b$  during one phase is always captured by
an $\{a,b\}$-oscillating cop.

\begin{lemma} \label{cycle}
Suppose that at his move, the robber $\mathcal R$ occupies a vertex $v$ of a cycle
$C$ of a graph $G$ and is not visible after this move. Then
$\mathcal R$ has a move (either staying at $v$ or going to
a neighbor of $v$) such that the cop does not capture the robber
during his next move.
\end{lemma}

\begin{proof}
Let $u$ be a neighbor of $v$ in $C$ which is not occupied by the cop. Since the robber
will not be visible after his next move, the
strategy of the cop is defined {\it a priori}. Let $z$ be the next
vertex to be occupied by the cop. Then the robber can stay at $v$
if $v\neq z$ or can move to $u$ if $u\neq z$.
\end{proof}

\begin{lemma} \label{oscillating_cop} If during one phase,  the cop
is performing $\{ a,b\}$-oscillating
moves and the robber moves to one of the vertices $a$ or $b,$ then
the robber is captured either immediately or at the next move of the
cop.
\end{lemma}

\begin{proof}
  Suppose that $\mathcal R$ moves to the
  vertex $a.$ If $\mathcal C$ is located at $a,$ then the robber is
  captured immediately. If $\mathcal C$ is located at $b$ and this is
  not the last vertex of the phase, then $\mathcal C$ will move to $a$
  and will capture there the robber. Finally, if $a$ and $b$ are the positions of
  $\mathcal R$ and $\mathcal C$ at the end of the phase, then the
  robber will be visible at $a$ and with the next visible move of $\mathcal C$ from $b$ to $a,$
  the robber will be caught at $a.$
\end{proof}

\subsection{On the inclusion of ${\mathcal C}{\mathcal W}{\mathcal W}(k+1)$ in ${\mathcal C}{\mathcal W}{\mathcal W}(k)$}
Clarke \cite{Clarke} noticed that for any $k\ge
2,$ the inclusion ${\mathcal C}{\mathcal W}{\mathcal F}{\mathcal
  R}(k)\subseteq {\mathcal C}{\mathcal W}{\mathcal W}(k)$ holds. 
Contrary to the classes considered in the previous
section which collapses for $k \geq 3$, we present now, for each $k$, an
example of a graph in ${\mathcal C}{\mathcal W}{\mathcal  W}(k)\setminus {\mathcal C}{\mathcal W}{\mathcal W}(k+1).$

\begin{proposition}\label{prop:stricte} For any $k\ge 2,$  ${\mathcal C}{\mathcal W}{\mathcal F}{\mathcal
  R}(k)$ is a proper subclass of ${\mathcal C}{\mathcal W}{\mathcal W}(k).$
For any  $k\ge 1,$ there exists a graph contained in
${\mathcal C}{\mathcal W}{\mathcal  W}(k)\setminus {\mathcal C}{\mathcal W}{\mathcal W}(k+1).$
\end{proposition}

\begin{proof}
To see the inclusion ${\mathcal C}{\mathcal W}{\mathcal F}{\mathcal
R}(k) \subseteq {\mathcal C}{\mathcal W}{\mathcal W}(k)$ (which was also mentioned in \cite{Clarke}),
it suffices to note that  we can interpret the moves at speed $k$ of the robber as if the cop moves only when the robber is visible
(i.e., each $k$th move).  Now, let $S_3$  be the 3-sun, the graph on 6 vertices obtained by gluing a triangle
to each of the three edges of another triangle (see Fig.~\ref{fig:examples}(a)). Since no vertex of $S_3$ has a maximum neighbor, the 3-sun is not dually chordal,
thus $S_3\notin \CW(2)$ by Theorem \ref{dually_chordal}. Then clearly, $S_3$ is not a big brother graph either. On the other hand, $S_3 \in {\mathcal
C}{\mathcal W}{\mathcal W}(k)$ for any $k\ge 2.$  Indeed,  initially the cop is placed at a vertex $u$ of degree
$4.$ Then, the robber shows himself  at the unique vertex $v$ which is not adjacent to $u.$
Let $x$ and $y$ be the two neighbors of $v$ in $S_3$. The
strategy of the cop consists in oscillating between $x$ and $y$
until the robber becomes visible again. Suppose without loss of generality that the
cop's sequence of moves is $x,y,x,y,\ldots,y.$ Then from Lemma \ref{oscillating_cop} we infer
that $\mathcal R$  is jammed at vertex $v.$  At the end,  when
the robber shows his position again, then
either he is at $v$ or he desperately moves to $x.$ In both cases,
he is caught by $\mathcal C$ at the next move. This shows that ${\mathcal C}{\mathcal W}{\mathcal F}{\mathcal
R}(k)$ is a proper subclass of  ${\mathcal C}{\mathcal W}{\mathcal W}(k)$

Now we will establish the second assertion.  Let $k\geq 1$ and $G_k$
be the graph defined as follows.  The vertex set of $G_k$ is
$\{x,y,u,v,u_1,\ldots,u_{k},v_1,\ldots,v_{k}\}$. The vertex $x$ is
adjacent to any vertex except $v,$ while $y$ is adjacent to any vertex
except $u$. For any $i<k$, the couples
$\{u_i,u_{i+1}\},\{u_i,v_{i+1}\}, \{v_i,v_{i+1}\}, \{v_i,u_{i+1}\}$
are edges of $G_k$.  Finally, $u$ is adjacent to $x,u_1,$ and $v_1$,
while $v$ is adjacent to $y,u_{k},$ and $v_{k}$ ($G_4$ is depicted in
Fig.~\ref{fig:examples}(b)). To prove that $G_k\in {\mathcal
  C}{\mathcal W}{\mathcal W}(k)$, consider the following strategy for
one cop. Initially, the cop occupies $x$. To avoid being caught
immediatly, the robber must show up at $v$. The cop occupies
alternatively $x$ and $y$ in such a way that after $k$ moves he is at
$y$ (if $k$ is odd, then the cop passes his first move). Therefore,
after $k$ steps, the robber shows up at a vertex of
$N_{k}(v,G\setminus\{x,y\}) \cup \{x\} \subseteq N_1(y),$ and at the
next move the cop caught him. On the other hand, we assert that in $G_k$ a
robber with witness $k+1$ can evade  against any strategy of
the cop.  Indeed, assume without loss of generality (in view of
symmetry) that the initial position of the cop belongs to the set
$L=\{x,u,u_1,\ldots,u_{\lceil k/2 \rceil},v_1,\ldots, v_{\lfloor k/2
  \rfloor}\}$. Then the robber chooses $v$ (or $v_{1}$ if $k=1$ and
the cop is occupying $u_{k}$) as his initial position.  Let $z$ be the
vertex occupied by the cop after $k+1$ steps. If $z \in L$, then by
Lemma \ref{cycle} the robber can move in the triangle $\{v,v_{k},y\}$
in order to avoid the cop during the $k+1$ steps and to finish at a
vertex of the triangle that is not adjacent to $z$. If $z \notin L$,
then the robber uses the $k+1$ steps to reach $u$ (or $u_1$ if $k=1$
and $z=v_1$). At any step, there is some $i\leq k$, such that the two
vertices $u_i$ and $v_i$ allow the robber to decrease his distance to
$u$ (or to $u_1$) by one; the robber chooses one of these vertices
that is not occupied and will not be occupied by the cop after his
move.
\end{proof}

\begin{figure}[t]
\centerline{
\begin{tabular}{ccc}
{\input{6sun.pstex_t}}& 
& {\input{kwitness.pstex_t}}\\ 
(a) & \qquad \qquad & (b)\\
\end{tabular}}
\caption{Two graphs in (a)  ${\mathcal C}{\mathcal
W}{\mathcal W}(k) \setminus {\mathcal C}{\mathcal W}{\mathcal
F}{\mathcal R}(k)$, $k\geq 2$ and (b)  ${\mathcal C}{\mathcal W}{\mathcal W}(4) \setminus
{\mathcal C}{\mathcal W}{\mathcal W}(5)$.}
\label{fig:examples}
\end{figure}

\medskip\noindent
{\bf Open question 2:} Is it true that ${\mathcal C}{\mathcal W}{\mathcal  W}(k+1)\subset {\mathcal C}{\mathcal W}{\mathcal W}(k)?$
\medskip

\subsection{$\bigcap_{k\ge 1} {\mathcal C}{\mathcal
  W}{\mathcal W}(k)$ and big two-brother graphs }

In analogy to  the big brother graphs, we say that a graph $G$ is
called a {\it big two-brother graph}, if $G$ can be represented as
an ordered union of subgraphs $G_1,\ldots,G_r$ in the form of a tree
$T$ rooted at $G_1$ such that (1) $G_1$ has a dominating vertex and
(2) any $G_i, i>1$, contains one or two adjacent vertices
disconnecting $G_i$ from its father and one of these two vertices
dominates $G_i$. Note that if $G_i$ and its father intersect in an
articulation point $x$, then $x$ is not necessarily the vertex which
dominates $G_i$. Equivalently, $G$ is a big two-brother graph if $G$
can be represented as a union of its subgraphs $G_1,\ldots,G_r$
labeled in such a way that $G_1$ has a dominating vertex, and for
any $i>1$, either the subgraph $G_i$ intersects $\cup_{j<i} G_j$ in
two adjacent vertices $x_i,y_i$ belonging to a common subgraph $G_j,
j<i,$ so that $y_i$ dominates $G_i$, or $G_i$ has a dominating
vertex $y_i$ and intersects $\cup_{j<i} G_j$ in a single vertex
$x_i$ (that may coincide with $y_i$); we will call such a
decomposition $G_1, \ldots, G_r$ a \emph{btb-decomposition} of $G$.
The vertices $y_i$ and $x_i$ are the big and the small brothers of
$G_i.$ Let ${\mathcal C}{\mathcal W}{\mathcal W}$ be the class of
all big two-brother graphs.  See Fig.~\ref{bigBrother}(b) for an
example of a big two-brother graph. As for big brother graphs, one
can associate a rooted tree $T$ with the decomposition $G_1, \ldots,
G_r$ of a big two-brother graph $G$. Obviously any big brother graph
$G$ is also a big two-brother graph because the required union of
subgraphs is provided by the block decomposition of $G$ and
$x_i=y_i$ is the articulation point of the block $G_i=B_i$ relaying
it with its father. The 2-trees and, more generally, the chordal
graphs in which all minimal separators are vertices or edges are
examples of  big two-brother graphs which are not big brother
graphs.

\begin{theorem}\label{theor:CWW}
A graph $G=(V,E)$ is $k$-winnable for all $k\ge 1$ if and only if $G$
is a big two-brother graph, i.e., ${\mathcal C}{\mathcal W}{\mathcal
  W}= \bigcap_{k\ge 1}{\mathcal C}{\mathcal W}{\mathcal W}(k).$
\end{theorem}

\begin{proof}
First we show that any big two-brother graph $G$ is $k$-winnable for
any $k\ge 1.$ Let $G_1,\ldots,G_r$ be a btb-decomposition of $G$. We
consider the following strategy for the cop. The cop starts the game
in the big brother of the root graph $G_1$ and, more generally, at the
beginning of each phase, we have the following property: the cop is
located in the big brother $y_i$ of some subgraph $G_i$ such that the
robber is located in a subgraph $G_k$ that is a descendent of $G_i$ in
the decomposition tree $T$ of $G$.  If $G_i = G_k$, then the cop will
capture the robber at the first move of the phase. Otherwise, let
$G_{j}$ be the son of $G_i$ on the unique path of $T$ between $G_i$
and $G_k$. If $G_i$ and $G_{j}$ intersect in an articulation point
$x_j$, then the cop moves from $y_i$ to $x_j$, stays there during
$k-2$ steps, and then, at the last step of the phase, if $x_j$ is not the
big brother $y_j$ of $G_j$, he moves to $y_j$. If $G_i$ and $G_{j}$
intersect in an edge $x_jy_j$ where $y_j$ is the big brother of $G_j$,
then the cop moves from $y_i$ to one of the vertices $x_j,y_j$ and
then oscillate between $x_j$ and $y_j$ in such a way that when
$\mathcal R$ becomes visible again $\mathcal C$ occupies the vertex $y_j$
(the decision to move first to $x_j$ or to $y_j$ depends only on the
parity of $k$).

During this phase, the robber cannot leave the subgraph induced by the
descendants of $G_j$, otherwise he has to go from $G_j$ to $G_i$. In
the first case, the cop stays during the whole phase in the unique
vertex $x_j$ which cannot be traversed by the robber. In the second
case, the cop oscillates between $x_j$ and $y_j$; therefore, by
Lemma~\ref{oscillating_cop} the robber cannot traverse
$\{x_j,y_j\}$. Therefore, after this phase, the invariant is preserved
and the distance in $T$ between the root and the subgraph $G_j$
hosting the cop has strictly increased.  Thus after at most diameter
of $T$ phases, $\mathcal R$ and $\mathcal C$ will be located in the
same subgraph $G_k$, and the cop captures the robber.

Conversely, let $G \in {\mathcal C}{\mathcal W}{\mathcal W}(k)$ for
any $k\ge 1$. If $G$ has a vertex $z$ of degree 1, then $G'=G\setminus
\{z\}$ is a retract of $G,$ thus $G'\in {\mathcal C}{\mathcal
  W}{\mathcal W}(k)$ for any $k\ge 1.$ Hence $G'$ has a
btb-decomposition $G_1,\ldots,G_{r-1}$ by induction hypothesis.  If
$w$ is the unique neighbor of $z,$ then setting $G_r$ to be the edge
$zw$ and $y_r=x_r:=w,$ we will conclude that $G$ is a big two-brother
graph as well. So, we can suppose that $G$ does not contain vertices
of degree 1.

Since $G \in {\mathcal C}{\mathcal W}{\mathcal W}(n^2),$ applying
Proposition \ref{prop:particularEdge} below for $k=n$, where $n$ is
the number of vertices of $G,$ we deduce that $G$ contains a vertex
$v$ and two adjacent neighbors $x,y$ of $v$ such that
$N_{n}(v,G\setminus\{x,y\}) \subseteq N_1(y)$. This means that the
connected component $C$ of $G\setminus\{x,y\}$ containing the vertex
$v$ is dominated by $y.$ The graph $G':=G(V\setminus C)$ is a retract
of $G,$ thus by Theorem 3 of \cite{Clarke} $G'\in {\mathcal
  C}{\mathcal W}{\mathcal W}(k)$ for any $k\ge 1.$ By induction
assumption, either $G'$ is empty or $G'$ has a btb-decomposition
$G_1,\ldots,G_{r-1}.$ If $G'$ is empty, then, since $y$ dominates $C,$
we conclude that $G$ has a btb-decomposition consisting of a single
subgraph. Otherwise, setting $G_r:=G(C\cup\{ x,y\})$, $y_r := y$ and
$x_r := x$, one can easily see that $G_1,\ldots, G_{r-1}, G_r$ is a
btb-decomposition of $G.$
\end{proof}

\begin{proposition} \label{prop:particularEdge}
Let  $G \in {\mathcal C}{\mathcal W}{\mathcal
W}(k^2)$ for $k\ge 1$. If the minimum degree of a vertex of  $G$ is at least $2$, then $G$ contains a vertex $v$ and an edge $xy$
such that $N_{k}(v,G\setminus\{x,y\})\subseteq N_1(y)$.
\end{proposition}

\begin{proof}
If $G$ contains a dominating vertex $y,$ then the result follows by
taking as $x$ any vertex of $G$ different from $y.$ Assume thus that
$G$ does not have any dominating vertex.  Consider a parsimonious
winning strategy of the cop and suppose that the robber uses a
strategy to avoid being captured as long as possible.  Since $G$ does
not contain leaves, the robber is caught immediately after having been
visible, i.e., at step $pk^2+1.$ Since $G$ does not have dominating
vertices, the robber is visible at least twice, i.e. $p\ge 1.$ Let $y$
be the vertex occupied by the cop when the robber becomes visible for
the last time before his capture. Let $v$ be the next-to-last visible
vertex occupied by the robber, i.e., his position at step
$(p-1)k^2+1,$ and let $c_0$ be the vertex occupied by the cop at that
moment.   Finally, let
$S^p_c=(c_0,c_1,\ldots,c_{k^2}=y)$ be the trajectory of the cop
between the steps $(p-1)k^2+1$ and $pk^2+1$ (repetitions are
allowed). Note that $v \notin N_1(c_0),$ otherwise the robber would
have been caught immediately at step $(p-1)k^2+1$. We distinguish two cases depending on whether or not the cop
occupies $y$ at least once every two consecutive steps.

\vspace{1ex}
\noindent
{\bf Case 1:} There exists an index  $(p-1)k^2+1\le i<pk^2-1$ such that $y \notin \{c_i,c_{i+1}\}$.

Let $i$ be the largest index satisfying the condition of Case 1 and
set $x:=c_{i+1}$.  We will use the following assertion.

\vspace{-0.5ex}
\begin{claim}\label{claim-cycle}
If $G$ contains a cycle $C$ and a vertex $w\in C$ such that $d(v,w)<
d(c_1,w)-1,$ then $G\setminus\{x,y\}$ has a connected component that
is dominated by $y$.
\end{claim}
\vspace{-0.5ex}

\begin{proof}
Let $w$ be a closest to $v$ vertex satisfying the condition of the claim. If the assertion of the
claim is not satisfied, we will exhibit a strategy allowing the robber to escape the cop during
more steps, contradicting the choice of the strategy of the robber.
Suppose that at the beginning of the $p$th phase the robber move from
$v$ to $w$ along a shortest $(v,w)$-path.  Since $d(v,w)<d(c_1,w)$,
the robber cannot be intercepted by the cop during these moves.
Suppose that the robber reaches the vertex $w$ before the $i$th step
when the cop arrives at $c_i$. Then by Lemma \ref{cycle} the robber
can safely move on $C$ until the cop reaches the vertex $c_i$.

Let $z$ be the position of $\cR$ when $\cC$ reaches
$c_i$.  Then $z \in N_1(y)$, otherwise the robber could stay at $z$
without being caught because starting with this step the cop moves
only on vertices of $N_1(y)$. Suppose that there exists a vertex $t$ at
distance $2$ from $y$ in $G\setminus\{x\}$. Let $r \neq x$ be a common
neighbor of $t$ and $y$. The following sequence of moves is valid for
the robber: when the cop is in $c_i$, the robber goes from $z$ to $y$
(or stays in $y,$ if $z = y$); once the cop has moved to $x = c_{i+1}$,
the robber goes from $y$ to $r$; finally, once the cop has moved to
$y$, the robber goes from $r$ to $t$. After this step, by definition
of $c_i$, the cop only stays in $N_1(y)$ and finishes in $y$. Hence,
the robber can remain in $t$ and will not be captured the next time he
shows up, a contradiction. This concludes the proof of the claim.
 \end{proof}

If the vertex $v$ belongs to a cycle $C,$ then setting $w:=v$ and applying Claim 1 we conclude
that  $y$ dominates the connected component of  $G\setminus \{x,y\}$
containing $v$,
establishing thus the assertion of Proposition
\ref{prop:particularEdge}.  So, suppose that $v$ is an articulation
point of $G$ not contained in a cycle. Since the minimum degree of
$G$ is at least $2$, $G\setminus \{ v\}$ has a connected component
$D$ that does not contain $c_0$ (nor $c_1$). Necessarily $D$
contains a cycle $C$, otherwise we will find in $D$ a vertex of
degree 1 in $G.$  Since any path from $c_1$ to a vertex $w$ of $C$
passes via $v$ and $c_1$ is not adjacent to $v,$  we obtain
$d(v,w)<d(c_1,w)-1$. The result then follows from the claim. This
concludes the analysis of Case 1.

\vspace{1ex}
\noindent
{\bf Case 2:} For any $(p-1)k^2 \leq i\leq pk^2$ we have $y \in \{c_i,c_{i+1}\},$ i.e., $\mathcal C$ occupies $y$ at
least once every $2$ steps.

First, assume that there exists a vertex $x$ (possibly $x=y$) and
$(p-1)k^2 \leq i\leq pk^2-k$ such that $c_i,\ldots,c_{i+k} \in
\{y,x\},$ i.e., that there are at least $k$ consecutive steps when
the cop remains at $x$ or $y$. Then, we claim that
$N_{k}(v,G\setminus\{x,y\}) \subseteq N_1(y)$. Indeed, pick $z\in
N_{k}(v,G\setminus\{x,y\})$ and let $P = (v= p_1, \ldots,p_k = z)$
be a shortest path in $G\setminus\{x,y\}$ between $v$ and $z$. Until
the $i$th step of the phase, the robber may progress ``slowly" along
$P$: either by staying at his current position, or moving to the
next vertex of $P$ toward $z$, depending on the moves of the cop.
The cop starts oscillating between $x$ and $y$ at step $i$. Then
during the next $k$ steps, the robber can follow $P$ until he
reaches $z$ (since the length of $P$ is at most $k$). Therefore, if
$z$ is not a neighbor of $y$, then the robber can remain at $z$
until step $k^2p$ without being captured. Since by our assumption
the robber is caught at step $k^2p$, necessarily $z \in N_1(y)$.
Hence $N_{k}(v,G\setminus\{x,y\}) \in N_1(y)$  and  the assertion of
Proposition \ref{prop:particularEdge} holds.

Therefore, we may assume that between the steps $(p-1)k^2$ and
$pk^2$, for all $k$ consecutive steps, the cop occupies at least
three distinct vertices (one of which is $y$). We assert in this
case that $N_{k}(v,G\setminus\{ y\}) \subseteq  N_1(y)$. Pick $z \in
N_{k}(v,G\setminus \{y\})$ and let $P$ be a shortest path between
$v$ and $z$ in $G\setminus \{ y\}$. Then for any vertex $w$ of $P$,
among any sequence of $k$ moves of the cop we can find three
consecutive moves during which the cop does not occupy $w$.
Therefore, for any sequence of $k$ consecutive steps the robber can
reduce by one his distance to $z$ by moving on $P$ towards $z$
without being captured. Hence, he will reach $z$ before step $pk^2.$
If $z$ is not adjacent to $y$, then staying at $z$ the robber will
not be captured, a contradiction. This concludes the proofs of
Proposition \ref{prop:particularEdge} and Theorem \ref{theor:CWW}.
\end{proof}


\section{Cop-win graphs for game with witness: classes
${\mathcal C}{\mathcal  W}{\mathcal W}(k)$}\label{sec-witnessII}

In this section we investigate the dismantling orders related to
$k$-winnable graphs.  We provide a dismantling order which must be
satisfied by all graphs of the class ${\mathcal C}{\mathcal
  W}{\mathcal W}(2)$. We show that this order is not sufficient but
some its reinforcement is. Then we continue with similar results about
$k$-winnable graphs for odd values of $k\ge 3.$

\subsection{Class ${\mathcal C}{\mathcal  W}{\mathcal W}(2)$}

We continue with the definition of a dismantling ordering which seems
to be intimately related with the witness variant of the cop and
robber game. Again, we will consider a slightly more general version
of the game: given a subset of vertices $X$ of a graph $G=(V,E),$ the
$X$-{\it restricted $k$-witness game} of cop and robber, is a
variant in which $\mathcal R$ can pass through any vertex of $G,$
$\mathcal C$ can move only inside $X$, and all visible positions of the
robber are at vertices of $X$. Then $X$ is called $k$-{\it winnable} if for
any starting positions of
$\mathcal C$ and $\mathcal R,$ the cop wins in the $X$-restricted
variant of the $k$-witness version of the game.  We will say that a
subset of vertices $X$ of a graph $G=(V,E)$ is $k$-{\it
  bidismantlable} if the vertices of $X$ can be ordered
$v_1,\ldots,v_m$ in such a way that for each vertex $v_i, 1\le i<m,$
there exist two adjacent or coinciding vertices $x,y$ with $y=v_j,
x=v_\ell$ and $j,\ell>i$ such that $N_{k}(v_i,G\setminus\{ x,y\})\cap
X_i\subseteq N_{1}(y),$ where $X_i:=\{ v_i,v_{i+1},\ldots,v_m\}$ and
$X_m=\{ v_m\}.$ We say that a graph $G=(V,E)$ is $k$-{\it
  bidismantlable} if its vertex-set $V$ is $k$-bidismantlable. In case
$k=2,$ the inclusion $N_{2}(v_i,G\setminus\{ x,y\})\cap X_i\subseteq
N_{1}(y),$ can be equivalently written as $N_2(v_i,G\setminus\{
x\})\cap X_i \subseteq N_1(y)$.  Any $(k,1)$-dismantlable graph is
$k$-bidismantlable but the converse is not true: for any $k\ge 2,$ the
3-sun $S_3$ presented in Fig.~\ref{fig:examples} is $k$-bidismantlable but
not $(k,1)$-dismantlable. In some proofs, we will denote by $x(v)$ and $y(v)$
the vertices eliminating a vertex $v$ in a $k$-bidismantling order.

\begin{proposition} \label{2witness}
Any graph $G=(V,E)$ of ${\mathcal C}{\mathcal W}{\mathcal W}(2)$  is $2$-bidismantlable.
\end{proposition}

\begin{proof}
  Suppose that a subset $X \subseteq V$ is $2$-winnable and assume
  that there exists an order $u_1, \ldots u_\ell$ on the vertices of
  $V\setminus X$ such that for each $1\le i\le \ell$, there exist the
  vertices $x(u_i), y(u_i) \in X_{i+1}$ such that $N_2(u_i,
  G\setminus\{x(u_i),y(u_i)\}) \cap X_i \subseteq N_1(y(u_i))$ holds,
  where $X_i = \{u_{i}, \ldots, u_\ell\} \cup X$. We show by induction
  on the size of $X$ that the set $X$ is $2$-bidismantlable. Assume
  $|X| \geq 2$, otherwise, $X$ is trivially $2$-bidismantlable.
  We first show that we can select a vertex $v_1\in X$, a vertex $y
  \in N(v_1)\cap X$, $y\neq v_1$, and a vertex $x \in N_1(y)\cap
  N(v_1)\cap X$ such that $N_2(v_1,G\setminus\{x,y\}) \cap X \subseteq
  N_1(y)$. If there exists a vertex $y \in X$ such that $X \subseteq
  N_1(y)$, then taking $x := y$ and any vertex of $X\setminus \{ y\}$
  as $v_1$, we are done. So, further we assume that $X$ does not
  contain dominating vertices.

  Consider a parsimonious winning strategy of the cop and a maximal
  trajectory of the robber.  First suppose that the capture happened
  when $\mathcal R$ is invisible. Let $v_1$ be the last position where the robber
  is visible. Let $a$ be the position of the cop when the robber shows
  up in $v_1$. We know that $v_1 \notin N(a)$, otherwise the cop would
  have captured the robber before. Let $y$ be the vertex where
  $\cC$ moves when he sees $\cR$ in $v_1$. Since the robber is
  captured when he is invisible, it implies he is captured in $v_1$.
  Moreover, since the robber follows a maximal trajectory, it implies
  that $N_2(v_1,G\setminus\{y\}) \cap X = \{v_1\}$, otherwise the
  robber could live longer. Consequently, by setting $x := y$, we have
  $N_2(v_1, G\setminus\{x,y\}) \cap X \subseteq N_1(y)$.

  Now suppose that $\mathcal C$ captures $\mathcal R$ at the next
  visible move. This means that when $\mathcal C$ sees $\mathcal R$,
  the cop is located in some vertex $y\in X$ and the robber is located
  in some vertex $w\in X$ and $w\in N_1(y)$ holds. Then the cop moves
  from $y$ to $w$ and captures $\mathcal R$ there. Denote by $v_1$ the
  vertex of $X$ where $\mathcal R$ is visible for the next-to-last
  time.  Suppose that after having seen the robber in $v_1,$ the cop
  moves first to a vertex of $X$ which we denote by $x$ and then to
  vertex $y$. Note that $x \neq v_1$ (otherwise the robber would have
  been caught when he shows up in $v_1$) and that $y$ may coincide
  with $x$ or with $v_1$. When the cop moves to $x$, the
  robber first moves to some vertex $u\in N_1(v_1)\setminus \{ x\}$
  and then, when $\cC$ moves to $y$, $\cR$ moves to a vertex $w\in
  N_1(u)\cap X\subseteq (N_2(v_1, G\setminus \{ x\})\cup\{ x\}).$ By
  the definition of the vertices $y$ and $w$, in $y$ the cop sees (for
  the last time) the robber which is located at $w$ and with the next
  move captures him. Since $\mathcal R$ follows a maximal sequence of
  moves before his capture, any vertex of $N_2(v_1, G\setminus \{
  x\})\cap X$ must be adjacent to $y,$ otherwise, if there exists
  $z\in N_2(v_1, G\setminus \{ x\})\cap X$ not adjacent to $y,$
  instead of moving to $w,$ in two moves the robber can safely reach
  $z$ and survive for a longer time. Thus $N_2(v_1,G\setminus\{
  x\})\cap X\subseteq N_1(y)$ holds.

  If $v_1 \neq y$, then we are done. If $v_1 = y$, then $N_2(y,G\setminus\{x\}) \cap X \subseteq
  N_1(y)$. If $N_1(y) \cap X \subseteq N_1(x)$, then $N_2(v_1,
  G\setminus \{x\})\cap X \subseteq N_1(y) \cap X \subseteq N_1(x)$
  and thus by setting $y(v_1) := x(v_1) := x$, we have $N_2(v_1,
  G\setminus\{x(v_1),y(v_1)\}) \cap X \subseteq N_1(y(v_1))$ and again we are done.  Suppose
  now that there exists a vertex $v \in N_1(y) \cap X$ which does not belong
  to $N_1(x)$. We assert that $N_2(v,G\setminus\{x,y\}) \cap X \subseteq
  N_1(y)$. Since $N_1(v, G\setminus\{x,y\}) \cap X \subseteq N_2(y,
  G\setminus\{x\}) \cap X \subseteq N_1(y)$, any neighbor $u$ of $v$ in $X$
  is a neighbor of $y$. Consider a vertex $u \in
  N_2(v,G\setminus\{x,y\}) \cap X$ and suppose there exists a vertex
  $r \in N_1(v) \cap N_1(u) \cap X \setminus\{x,y\}$. Then $r \in N_1(y)$
  and thus $u \in N_2(y, G\setminus\{x\}) \cap X \subseteq
  N_1(y)$. Suppose now that there does not exist any vertex $r \in
  N_1(v) \cap N_1(u) \setminus\{x,y\}$ that belongs to $X$. Among all
  vertices in $N_1(v) \cap N_1(u) \setminus\{x,y\}$, let $r$ be the last vertex
  occurring in the ordering $u_1, \ldots,
  u_{\ell}$. Then, since $u,v \in N_1(r) \cap X$, $u,v \in N_1(y(r))$
  and consequently, $y(r) \neq x$, since $v \notin N_1(x)$. By our
  choice of $r$, we know that $y(r) \in X$ and thus there exists a
  vertex in $N(v) \cap N(u) \cap X \setminus\{x,y\}$, a contradiction.
  Therefore, by setting $x(v):=y(v):=y,$ we have $N_2(v,
  G\setminus\{x(v),y(v)\}) \cap X \subseteq N_1(y(v)).$
  In the rest of the proof, we denote by $v_1$ the vertex satisfying
  this condition, it can be either $v_1$ or $v$.

  Consider the set $X':=X\setminus\{ v_1\}$. Note that
  $V\setminus X' = V\setminus X \cup \{v_1\}$, and there exists an
  order $u_1, \ldots u_\ell, u_{\ell+1}:=v_1$ on the vertices of
  $V\setminus X'$ such that for each $1\le i\le \ell+1$, there
  exist $x(u_i), y(u_i) \in X_{i+1}$ such that $N_2(u_i,
  G\setminus\{x(u_i),y(u_i)\}) \cap X_i \subseteq N_1(y(u_i)).$ We show that the set
  $X'$ is 2-winnable as well.  Consider a positional
  parsimonious winning strategy $\sigma$ of the cop in $X$.  For any
  positions $c$ of the cop and $r$ of the robber in $X'$, we note
  $\sigma(c,r) = (c_1,c_2)$. As in the proof of Theorem~\ref{copwin},
  we construct a strategy that uses one bit of memory $m$: it is a
  function that associates to each $(c,r,m)$ a couple $((c_1',
  c_2'),m)$. As in the proof of Theorem~\ref{copwin}, the intuitive
  idea is that the cop plays using $\sigma$, except when he is in $y$
  and his memory contains $1$; in that case, he plays using $\sigma$
  as if he was in $v_1$.

  If $m = 0$ or $c \neq y$, let $(c_1, c_2) = \sigma(c,r)$. If $c_1 =
  v_1$, then $c'_1 = y$ and $c'_1 = c_1$ otherwise.  If $c_2 = v_1,$
  then $\sigma'(c,r,m) = ((c'_1,y),1)$ and $\sigma'(c,r,m) =
  ((c'_1,c_2),0)$ otherwise. If $m = 1$ and $c = y$, let $(c_1, c_2) =
  \sigma(v_1,r)$. If $c_1 = v_1$, then $c'_1 = y$ and $c'_1 = c_1$
  otherwise. If $c_2 = v_1$, then $\sigma'(y,r,1) = ((c'_1,y),1)$ and
  $\sigma'(y,r,1) = ((c'_1,c_2),0)$ otherwise. Since $N_1(v_1) \cap X
  \subseteq N_1(y)$, one can easily check that $\sigma'$ is a valid
  strategy for the $X'$-restricted game.

  By way of contradiction, suppose now that there exists an infinite
  $X'$-valid sequence $S_r'$ of moves of the robber in the
  $X'$-restricted game allowing him to escape forever against a cop
  using the strategy $\sigma'$.  First note that the sequence of moves
  $S_c$ of the cop playing $\sigma$ against $S_r'$ differs from the
  sequence of moves $S_c'$ of the cop playing $\sigma'$ against $S_r'$
  only in the positions where the cop is in $v_1$ in $S_c$.

  We show that there exists an infinite sequence $S_r$ in the
  $X$-restricted game enabling the robber to escape
  forever against a cop using the strategy $\sigma$. The visible
  positions of $\mathcal R$ in $S_r$ will coincide with the visible
  positions of $\mathcal R$ in $S_r'$ (thus the cop's strategies
  $\sigma$ and $\sigma'$ behave in the same way against both
  sequences). It is sufficient to show that if during a phase of
  $S_r'$, the robber goes from $r_0' \in X'$ to $r_2' \in X'$ via
  $r_1' \in V(G)$, then in the $X$-restricted game
  where the cop plays with strategy $\sigma$ (going first to $c_1$ and
  then to $c_2$), there exists $r_1$ such
  that $\mathcal R$ can go from $r_0'$ to $r_2'$ via $r_1$ without
  being captured in $r_1$.

  If $r'_1 \neq v_1$ or if $v_1 \notin \{c_1,c_2\},$ then one can choose
  $r_1 = r'_1$ (since $r_0', r_2' \in X'$, they are different from
  $v_1$). Thus, we may assume that $r_1' = v_1$ and that $c_1 = v_1$
  or $c_2 = v_1$. If $c_2 \in \{v_1,y\}$, then $c'_2 = y$. Since $r_1'
  = v_1$, $r'_2 \in N_1(v_1)\cap X \subseteq N_1(y)$ and thus
  the robber is captured when he shows up in $r'_2$, i.e., $S_r'$ does
  not enable the robber to escape forever. Consequently, $c_2 \notin
  \{v_1,y\}$ and $c_1 = v_1.$ In this case,
  $(r_0',r_1 := y,r_2')$ is a $X$-valid sequence since $r'_0, r'_2 \in
  N_1(v_1)\cap X \subseteq N_1(y)$ and moreover $y \notin
  \{c_1,c_2\}$ (since $c_1 = v_1$ and $y \neq c_2$). It implies that
  there exists an infinite $X$-valid sequence $S_r$ enabling the
  robber to escape forever, a contradiction.

  Starting from a positional strategy for the $X$-restricted game, we
  have constructed a winning strategy using memory for the
  $X'$-restricted game. As mentioned in the introduction, it implies
  that there exists a positional winning strategy for the
  $X'$-restricted game. Consequently, the set $X':=X\setminus\{ v_1\}$
  is 2-winnable as well. By induction assumption, $X'$ admits a
  2-bidismantling order $v_2,\ldots,v_m$. Then clearly
  $v_1,v_2,\ldots,v_m$ is a 2-bidismantling of $X.$ If $G$ is
  2-winnable, then its set of vertices is 2-winnable and therefore
  2-bidismantlable, showing that $G$ is 2-bidismantlable.
\end{proof}

We continue with two examples. The first one shows that we cannot
replace in the definition of 2-bidismantlability the condition
$N_{2}(v_i,G\setminus\{ x\})\cap X_i\subseteq N_{1}(y)$ by a weaker
condition $N_{2}(v_i,G_i\setminus\{ x\})\subseteq N_{1}(y)$ (i.e.,
instead of all vertices of $X_i$ reachable from $v_i$ by paths of
length 2 avoiding $x$ of the whole graph $G$ to consider only the
vertices reachable by such paths of the subgraph $G_i$). The second
example shows that unfortunately 2-bidismantlability is not a
sufficient condition.

\begin{proposition} Let $G$ be the graph from Fig.
  \ref{fig-contre-ex1}. Then $G$ admit a dismantling order satisfying the
  condition $N_{2}(v_i,G_i\setminus\{ x\})\subseteq N_{1}(y),$ however
  $G$ is not 2-bidismantlable nor 2-winnable.
\end{proposition}

\begin{figure}[h!]
\begin{center}
{\input{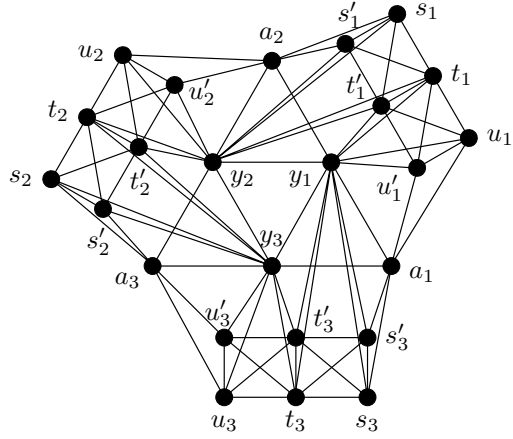}}
\caption{A weakly $2$-bidismantlable graph that is not
  $2$-bidismantlable}\label{fig-contre-ex1}
\end{center}
\end{figure}

\begin{proof}
Consider the following order on the vertices of $G$: $a_1$, $a_2$,
$a_3$, $u_1$, $u_1'$, $u_2'$, $u_3$, $u_3'$, $s_1$, $s_1'$, $s_2$,
$s_2'$, $s_3$, $s_3'$, $t_1$, $t_1'$, $t_2$, $t_2'$, $t_3'$, $t_3'$,
$y_1$, $y_2$, $y_3$. For each vertex $v \in V(G)\setminus\{y_3\}$,
we give below two adjacent vertices $x(v), y(v)$ that are eliminated
later than $v$ and such that $N_2(v,G_i\setminus\{x(v),y(v)\})
\subseteq N_1(y(v))$ (the vertex $x(v)$ is not defined if
$N_2(v,G_i\setminus\{y(v)\}) \subseteq N_1(y(v))$).


$$
\begin{array}{|c|cccccccccccc|}
\hline
v & a_1 & a_2 & a_3 & u_1 & u_1' & u_2 & u_2' & u_3 & u_3' & s_1 &
s_1' & s_2 \\
\hline
y(v) & y_1 & y_2 & y_3 & t_1 & t_1 & t_2 & t_2 & t_3 & t_3 & y_2 & y_2 &
y_3 \\
\hline
x(v) & y_3 & y_1 & y_2 & y_1 & y_1 & y_2 & y_2 & y_3 & y_3 & - & -
& - \\
\hline
\hline
v & s_2' & s_3 & s_3' & t_1 & t_1' & t_2 & t_2' & t_3 & t_3'
& y_1 & y_2 & y_3\\
\hline
y(v) & y_3 & y_1 & y_1 & y_2 & y_2 & y_3 & y_3 & y_1 & y_1 & y_2 & y_3
& -\\
\hline
x(v) & - & - & - & y_1 & y_1 & - & - & - & - & - & - & - \\
\hline
\end{array}
$$






\medskip
We prove now that $G$ is not $2$-bidismantlable. Note that for $a_1$
(resp. $a_2, a_3$), there exist $y(a_1) = y_1$ (resp. $y(a_2) = y_2,
y(a_3) = y_3$) and $x(a_1) = y_3$ (resp. $x(a_2) = y_1, x(a_3) = y_2$)
such that $N_2(a_1,G\setminus\{y_1,y_2\}) \subseteq N_1(y_1)$ (resp.
$N_2(a_2,G\setminus\{y_2,y_3\}) \subseteq N_1(y_2)$,
$N_2(a_3,G\setminus\{y_3,y_1\}) \subseteq N_1(y_3)$). Consequently,
any $2$-bidismantling order of $G$ can start with $a_1, a_2, a_3$. In fact,
one can check that any $2$-bidismantling of $G$ must start with a
permutation of $a_1,a_2,a_3.$ We will show now that it is impossible
to extend a $2$-bidismantling order starting with $a_1,a_2,a_3$.
To prove this, it suffices  to show that for any $v \in V(G) \setminus\{a_1,a_2,a_3\}$
and for all adjacent vertices $x(v), y(v) \in N_1(v)$, there exists
a vertex $z(v) \in N_2(v,G\setminus\{x(v),y(v)\})\setminus \{a_1,a_2,a_3\}$
such that $z(v)\notin N_1(y(v))$.  In view of symmetry of $G$, it is sufficient
to check this property for $v \in \{u_1, t_1, y_1\}$.

If $v = u_1$, then $y(u_1), x(u_1) \in \{a_1, u_1', t_1, t_1',
y_1\}$. If $y(u_1) \in \{a_1,u'_1,y_1\}$, then either $t_1 \neq
x(v_1)$, or $t_1' \neq x(v_1).$ In both cases, $s_1 \in
N_2(v,G\setminus\{x(u_1),y(u_1)\})$ and $s_1 \notin N_1(y(u_1))$. By
symmetry, we can suppose that $y(u_1) = t_1$. Since $a_1 \notin N_1(t_1)$, we must have
$x(u_1) \neq a_1$ and consequently, $s_3 \in
N_2(v,G\setminus\{x(u_1),y(u_1)\})$ and $s_3 \notin N_1(y(u_1))$.

If $v = t_1$, then $y(t_1), x(t_1) \in \{u_1, u_1', s_1, s_1', t_1',
y_1, y_2\}$. If $y(t_1) \in \{y_1,u_1,u_1'\}$
(resp. $y(t_1)=\{y_2,s_1,s_1'\}$), then set $z(v) = s_1$ (resp. $z(v)
= u_1$); in all cases, $z(v) \in N_2(v,G\setminus\{x(t_1),y(t_1)\})$
and $z(v) \notin N_1(y(t_1))$.  If $y(t_1) = t_1'$, then either
$x(t_1) \neq y_2$ or $x(t_1) \neq y_1$; in both cases, $y_3 \in
N_2(v,G\setminus\{x(t_1),y(t_1)\})$ and $y_3 \notin N_1(y(t_1))$.

If $v = y_1$, since $N_1(y_1) \subseteq N_1(y(y_1))$, the vertex $y(y_1)$
must belong to $N_1(y_1) \cap N_1(y_2) \cap N_1(y_3)$. Consequently,
by symmetry, we can assume  that $y(y_1)=y_2$. However, since $u_1 \in
N_1(y_1)\setminus N_1(y_2)$, we obtain $u_1 \in
N_2(y_1,G\setminus\{x(y_1),y(y_1)\})$ and $u_1 \notin N_1(y(y_1))$. This completes the proof that
$G$ is not $2$-bidismantlable. Since any graph $G \in \CWW(2)$ is $2$-bidismantlable,
it also implies that $G  \notin \CWW(2)$.
\end{proof}

\begin{proposition} Let $G$ be the graph from Fig. \ref{fig-contre-ex2}. Then $G$ is 2-bidismantlable,
however $G\notin \CWW(2).$
\end{proposition}

\begin{figure}[h!]
\begin{center}
{\input{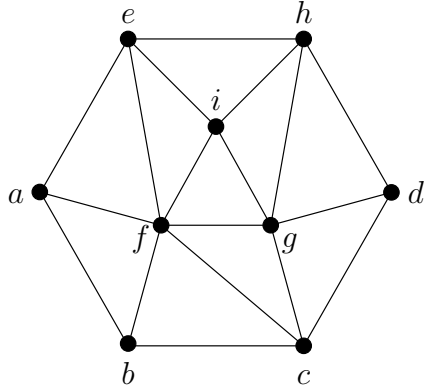}} \caption{A
$2$-bidismantlable graph $G\notin \CWW(2)$} \label{fig-contre-ex2}
\end{center}
\end{figure}

\begin{proof}
The graph presented in Fig.~\ref{fig-contre-ex2} is
$2$-bidismantlable with the following $2$-bidismantling order $a, b, c, d,
e, f, g, h, i,$ where each vertex $v$ is eliminated by the vertices
$x(v)$, $y(v)$ defined as follows:
$$
\begin{array}{|c|ccccccccc|}
\hline
v & a & b & c & d & e & f & g & h & i \\
\hline
y(v) & f & f & g & g & i & i & i & i & - \\
\hline
x(v) & e & c & f & h & - & - & - & - & -\\
\hline
\end{array}
$$
\medskip

However, one can show that for any vertex $c$  there exists a
vertex $r$ such that, if at step $i \geq 0$ the cop moves to (or
starts in) $c$ (going through any intermediate vertex), then the
robber can move to (or starts in) $r$ without being caught. Since for
any such couple $(c,r)$ the vertices $c$ and $r$ are not adjacent,
it means that the cop cannot catch the robber in this graph. The definition
of the pairs $(c,r)$ is given in the following table:

\medskip
$$
\begin{array}{|c|ccccccccc|}
\hline
c & a & b & c & d & e & f & g & h & i \\
\hline
r & d & d & e & e & b & d & e & b & b \\
\hline
\end{array}
$$
\medskip


Note that if the robber wants to go from $d$ to $e$, (resp., from $e$ to $d$),
then this means that the cop is in $a, b$ or $f$
(resp., wants to go to $a, b$ or $f$). Since $h \notin N(a)\cup
N(b)\cup N(f)$, $h$ cannot be the intermediate vertex used by the
cop. Thus, the robber can always go from $d$ to $e$ (resp. from $e$ to
$d$) via $h$.

If the robber wants to go from $b$ to $d$ (resp., from $d$
to $b$), then this implies that the cop is in $e, h, i$ (resp., wants to go
to $e, h, i$). Since $c \notin N(e)\cup N(h)\cup N(i)$, $c$ cannot be
the intermediate vertex used by the cop.  Thus, the robber can always
go from $b$ to $d$ (resp. from $d$ to $b$) through $c$.

If the robber wants to go from $b$ to $e$ (resp. from $e$ to $b$),
then this means that the cop neither starts in $a$ nor $f$ (because in this
case the robber would have been in $d$), nor goes to $a$ or $f$
(since in this case, the robber wants to go in $d$). Moreover, the
intermediate vertex used by the cop is different from $a$ or $f$.
In the first case (resp. second case), the robber
can go from $b$ to $e$ via $a$ (resp. $f$).
\end{proof}

We continue with a condition on 2-bidismantling which turns out to
be sufficient for 2-winability. We say that a graph $G$ is
\emph{strongly
  2-bidismantlable} if $G$ admits a 2-bidismantling order such that
for any vertex $v_i, i<n$, $y(v_i)=x(v_i)$ or $N_{2}(v_i,G\setminus\{
y(v_i)\})\cap X_i\subseteq N_{2}(x(v_i),G\setminus\{ y(v_i)\})$
(recall that $x(v)$ and $y(v)$ denote the vertices eliminating a
vertex $v$ in a 2-bidismantling order).


\begin{proposition} If a graph $G$ is strongly 2-bidismantlable, then
  $G\in \CWW(2)$.
\end{proposition}

\begin{proof} Suppose that a subset $X$ of vertices of $G$ admits a
strong 2-bidismantling order $v_1,\ldots,v_m$. Assume by induction
assumption that the set $X'=\{ v_2,\ldots,v_n\}$ is 2-winnable and
we will establish that the set $X$ itself is 2-winnable. Let
$N_2(v_1,G\setminus\{ x\})\cap X\subseteq N_1(y).$ Let $\sigma'$ be
a parsimonious positional winning strategy for $\mathcal C$ in $X'$.
We define the strategy $\sigma$ for $\mathcal C$ in $X$ as follows:
$\sigma(c,r)=r$ if $r \in N_1(c)$, $\sigma(c,v_1)=(x,y)$ if $c \in
N_1(x)$ (in this case, the robber will be caught during the next
move because $N_2(v_1,G\setminus \{x\}) \cup X \subseteq N_1(y)$)
and $\sigma(c,v_1)=\sigma'(c,x)$ otherwise, and
$\sigma(c,v)=\sigma'(c,v)$ in all other cases. We now prove that
$\sigma$ is winning. Let $S_r=(r_1,r_2,\ldots)$ be any $X$-valid
sequence of moves of the robber. We will transform $S_r$ into a
$X'$-valid sequence $S'_r=(r'_1,r'_2,\ldots)$ of moves of the robber
and prove that, since $\mathcal C$ playing $\sigma'$ eventually
captures $\mathcal R$ following $S'_r$, then $\mathcal C$ playing
$\sigma$ captures $\mathcal R$ following $S_r$.

 Let $r'_1:=x$ if $r_1=v_1$ and $r'_1:=r_1$ otherwise. Suppose that
$r'_1,\ldots,r'_{2j-1}$ ($j\geq 1$) have been already defined and we
wish to define $r'_{2j}$ and $r'_{2j+1}.$ We set
$r'_{2j+1}:=r_{2j+1}$ if $r_{2j+1}\ne v_1$ and $r'_{2j+1}:=x$
otherwise (indeed, when the cop sees the robber in the vertex $v_1,$
then $\mathcal C$ will plays against $\mathcal R$ as like the latter
was in $x$). We set $r'_{2j}:=r_{2j}$ in all cases unless $v_1 \in
\{r_{2j-1},r_{2j+1}\}$ and $r_{2j} \notin N_1(x)$ (in particular
$r_{2j}\neq y$). If $r_{2j-1}=v_1$  (resp., if $r_{2j+1}=v_1$) and
$r_{2j}\notin N_1(x)$, then there exists a common neighbor $u$ of
$r_{2j-1}$ (resp.,  $r_{2j+1}$) and $x$ different from $y$. The
choice of $r'_{2j}$ depends of the current position $c_{2j}$ of the
cop pursuing $\mathcal R$.  We set $r'_{2j}:=u$ if $c_{2j} \neq u$
and $r'_{2j}:=y$ otherwise (this is to avoid to artificially create
a move where the robber goes to a vertex occupied by the cop). It
can be easily seen that $S'_r$ is a $X'$-valid sequence of moves of
the robber.

Let $S'_c=(c'_1,c'_2,\ldots)$ be the $X'$-valid sequence of moves of the cop playing $\sigma'$ against a robber $\mathcal R'$ moving according to $S'_r$, and let $S_c=(c_1,c_2,\ldots)$ be the $X$-valid sequence of moves of the cop playing $\sigma$ against the robber $\mathcal R$ following $S_r$. It is easy to check that $S'_c$ and $S_c$ are similar
except one or two steps before the capture of the robber. Moreover, since $\sigma'$ is a winning strategy in $X'$, there is $j>0$ such that $c'_j=r'_j$.

First suppose that $\mathcal C$ captures the robber ${\mathcal R}'$
when he is visible, say ${\mathcal R}'$ is located in $r'_{2j+1}.$
If $r'_{2j+1}=r_{2j+1},$ then we are done. So, suppose that
$r'_{2j+1}\ne r_{2j+1},$ i.e., $r_{2j+1}=v_1$ and $r'_{2j+1}=x.$
Therefore,  when $\mathcal C$ sees ${\mathcal R}$ in $v_1,$  the cop
is located in a neighbor of $x$. According to $\sigma,$ $\mathcal C$
will move to $x$ and then to $y$, while $\mathcal R$ can only reach
a vertex in  $N_2(v_1,G\setminus \{x\}) \cap X$. Since
$N_2(v_1,G\setminus \{ x\})\cap X\subseteq N_1(y)$, the cop will
capture the visible robber at his next move.

Now suppose that $\mathcal C$ captures  ${\mathcal R}'$ when the latter  is invisible, say
${\mathcal R}'$ is located in $r'_{2j}.$ Again, if $r'_{2j}=r_{2j},$
then we are done. Otherwise, according to the definition of $S'_r,$ we
conclude that $r_{2j}$ is a common neighbor  of $r_{2j-1}$ and
$r_{2j+1}$ different from $y$  with either $v_1 = r_{2j+1}$ or $v_1=r_{2j-1}.$ Suppose that $v_1=r_{2j+1}$
(the other case is analogous), $r'_{2j}$ is either $y$ or
a common neighbor $u$ of $r_{2j-1}$ and $x$ provided by the strong
2-bidismantling order.  Since, between $r_{2j-1}$ and $r_{2j+1}=v_1$ the trajectory of ${\mathcal R}'$
avoids the cop if possible, we deduce that  $\{c_{2j-1},c_{2j}\} = \{u,y\}$ or $\{c_{2j},c_{2j+1}\} = \{u,y\}.$
If $\{c_{2j-1},c_{2j}\} = \{u,y\},$ then, when $\mathcal C$ sees
${\mathcal R}$ in $r_{2j-1},$  the cop is located in a neighbor of $r_{2j-1}.$  By the definition of $\sigma,$
$\mathcal C$ will move to $r_{2j-1}$ and captures $\mathcal R.$
Otherwise, if $\{c_{2j},c_{2j+1}\} = \{u,y\},$ then when the cop
sees ${\mathcal R}$ in $v_1,$ $\mathcal C$ is located in a neighbor of $x.$  By the definition of $\sigma,$
as before, $\mathcal C$ will move to $x$ and then to $y$, while $\mathcal R$
can only reach a vertex in $N_2(v_1,G\setminus \{x\}) \cup X$. Since
$N_2(v_1,G\setminus \{ x\})\cap X\subseteq N_1(y)$, the cop will
capture the visible robber at his next move.
\end{proof}

We conclude this section by showing that the existence of a strong
2-bidismantling order is not necessary.

\begin{figure}[h!]
\begin{center}
{\input{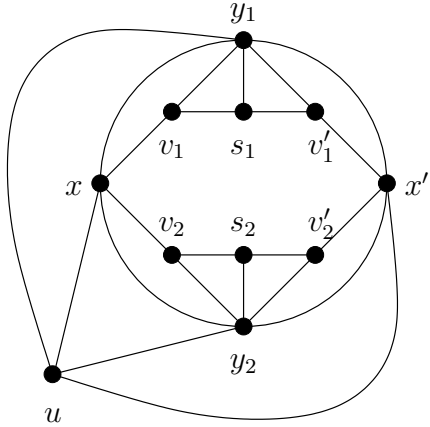}}
\caption{A graph $G \in \CWW(2)$ which is not strongly
  2-bidismantlable.}
\label{fig-contre-ex4}
\end{center}
\end{figure}

\begin{proposition} The graph $G$ from Fig.~\ref{fig-contre-ex4}
  belongs to $\CWW(2)$, however $G$ does not admit a strong
  2-bidismantling order.
\end{proposition}

\begin{proof}
We first show that the graph $G$ from Fig.~\ref{fig-contre-ex4} is
in $\CWW(2).$  The cop starts in
$u$. Hence, if the robber starts in $x, x', y_1$ or $y_2$, he is
immediately caught. If the robber starts in  $s_1$
(or $s_2$), then the cop moves to $y_1$ (resp., to $y_2$) and since
$N_2(s_1) \subseteq N_1(y_1)$ (resp., $N_2(s_2) \subseteq N_1(y_2)$),
the robber is caught the next time he shows up. If the robber starts
in $v_1$ (the cases  $v_1'$, $v_2$,
$v_2'$ are similar), then the cop first moves to $x$ and then to
$y_1$. Then the robber has to show up in a vertex of $\{v_1, s_1, v_1', x\}
\subseteq N_1(y_1)$ and the cop can catch him.

Consider now any $2$-bidismantling order of $G$.  Let $v$ be the
first vertex in this order which is different from $s_1, s_2.$ We
may assume without loss of generality that $v \in\{v_1,x,y_1,u\}$.
Let $X = V(G)\setminus\{s_1,s_2\}$.  Since there does not exist $t$
such that the set $N_1(x) \cap X$ (resp. $N_1(y_1) \cap X$ or
$N_1(u) \cap X$) is included in $N_1(t)$, it implies $v = v_1$. We
know that $x(v_1), y(v_1) \in N_1(v_1)$ and that $N_1(v_1) \cap X
\subseteq N_1(y(v_1))$.  Consequently, $y(v_1) \in \{x,y_1\}$. If
$y(v_1) = x$, then $x(v_1) \neq s_1$ (since $s_1$ and $x$ are not
adjacent) and thus $v_1' \in
N_2(v_1,G\setminus\{x(v_1),y(v_1)\})\setminus N_1(y(v_1))$, which is
impossible. Thus, $y(v_1) = y_1$. If $x(v_1) \neq x$, then $v_2 \in
N_2(v_1,G\setminus\{x(v_1),y(v_1)\})\setminus N_1(y_1)$.
Consequently, $y(v_1) = y_1$ and $x(v_1) = x$. However, $v_1' \in
N_2(v_1) \cap X$ but $v_1' \notin N_2(x,G\setminus\{y_1\})$ and thus
$G$ is not strongly 2-bidismantlable.
\end{proof}

\subsection{Classes ${\mathcal C}{\mathcal  W}{\mathcal W}(k)$ for $k\ge 3$}

In this subsection we show that $k$-bidismantlable graphs are
$k$-winnable for any odd $k \geq 3$. We also show that for any $k \geq
3$, there exist graphs in $\CWW(k)$ that are not $k$-bidismantlable,
i.e., for $k\ge 3,$ $k$-bidismantlability of a graph is not a
necessary condition to be $k$-winnable.

\begin{theorem}
For any odd integer $k\ge 3$, if a graph $G$ is $k$-bidismantlable, then $G \in {\mathcal C}{\mathcal W}{\mathcal
W}(k).$
\end{theorem}

\begin{proof}
Suppose that $X \subseteq V$ is a $k$-bidismantlable set of vertices of a graph
$G$.  We prove that there is a winning strategy for the cop in the
$X$-restricted $k$-witness game on $G$. To do so, we proceed as
in the papers \cite{IsKaKh,NowWin} and use the $k$-bidismantling order to mark all
$X$-configurations $(c,r)$. A {\it $X$-configuration} of  $X$-restricted game
is a couple $(c,r)$ that consists of a position of the cop $c \in X$ and a
position of the robber $r \in X$, with $r \neq c$. A $X$-configuration
$(c,r)$ is called {\it terminal} if $r \in N_1(c)$.

To mark the $X$-configurations, we use the following procedure {\tt Mark($X$)}.
\begin{enumerate}
\item Initially,
all $X$-configurations are unmarked.
\item Any {\it terminal} $X$-configuration $(c,r)$ is marked with label $1$.
\item
While it is possible, mark an unmarked $X$-configuration $(c,r)$ with
the smallest possible integer $\ell+1$ such that there exist vertices $y_{(c,r)} \in N_1(c) \cap
X$ and $x_{(c,r)}
\in (N_1(y_{(c,r)}) \setminus \{r\}) \cap X$ such that for all $z \in N_k(r,G\setminus
\{x_{(c,r)},y_{(c,r)}\}) \cap X$, the $X$-configuration $(y_{(c,r)},z)$ is marked with a label at most $\ell$.
\end{enumerate}
\begin{claim}
If all $X$-configurations are marked by {\tt Mark($X$)}, then
there is a winning strategy for the cop in the $X$-restricted
$k$-witness game on $G$.
\end{claim}

Indeed, pick any initial positions $c\in X$ of the cop and $r\in X$ of
the robber. If the configuration $(c,r)$ is terminal, then $r\in
N_1(c)$ and the robber is captured at the next move. Otherwise, the
cop first moves to $y_{(c,r)}$ and then oscillates between
$x_{(c,r)}$ and $y_{(c,r)}$ during $k-1$ steps, i.e., the cop ends in
$y_{(c,r)}$ since $k$ is odd. If during one of his invisible moves the
robber goes to $x_{(c,r)}$ or $y_{(c,r)},$ then he will be captured
immediately. Otherwise, in $k$ moves the robber goes from $r$ to a
vertex $z\in N_k(r,G\setminus \{ x_{(c,r)},y_{(c,r)}\})\cap X.$ According to {\tt Mark($X$)},
the label of $(y_{(c,r)},z)$ is strictly less than that of $(c,r)$.
Therefore, by repeating the same process, after a
finite number of steps either the cop captures the robber during an
invisible move or the cop and the robber arrive at a terminal
configuration.

\begin{claim}
If $X$ is $k$-bidismantlable, then {\tt Mark($X$)} marks all $X$-configurations.
\end{claim}

The general idea
of our proof follows  the proof of Theorem 12 of \cite{IsKaKh}.
Let $\{v_1,\ldots,v_t\}$ be a $k$-bidismantling ordering of $X$. We
prove by induction on $t-i$ that {\tt Mark($X_i$)} marks all
$X_i$-configurations, where $X_i=\{v_i,\ldots,v_t\}$. The assertion
trivially holds for $X_{t-1}$. Let $i<t-1.$ Assuming that all
$X_{i+1}$-configurations are marked by {\tt Mark($X_{i+1}$)}, we prove
that {\tt Mark($X_i$)} marks all $X_i$-configurations.

By definition of the $k$-bidismantling ordering, there exist two
adjacent or coinciding vertices $x,y \in X_{i+1}$ such that
$N_k(v_i,G\setminus \{x,y\}) \cap X_i \subseteq N_1(y)$.
Roughly speaking, {\tt Mark($X_i$)} marks the $X_i$-configurations in the same
order as {\tt Mark($X_{i+1}$)} marks the $X_{i+1}$-configurations, but
once a configuration $(c,y)$ with $c\in X_{i+1}$ is marked, {\tt
Mark($X_i$)} also marks the configuration $(c,v_i)$. Once {\tt
Mark($X_i$)} has marked all $X_i$-configurations $(c,r) \in X_{i+1}
\times X_i$, the remaining $X_i$-configurations $(v_i,r)$ with $r \in
X_{i+1}$ can also be marked by {\tt Mark($X_i$)}.

Let $\ell\geq 1$. By induction on $\ell$, we prove that {\it any
$X_{i+1}$-configurations $(c,r)$ that is marked by {\tt
Mark($X_{i+1}$)} with label at most $\ell$ will be also marked by {\tt
Mark($X_i$)}.} Moreover, if $r=y$, we prove that {\it once {\tt Mark($X_i$)}
has marked $(c,r)$, then it can mark $(c,v_i)$.} Let us first prove this assertion
for $\ell=1$. For any $(c,r) \in X_i \times X_i$ with
$r \in N_1(c)$, $(c,r)$ is marked by {\tt Mark($X_i$)} with label $1$ . If
$(c,y)$ is marked with label $1$ (i.e., $y \in N_1(c) \cap X_i$), then
$(c,v_i)$ can be marked with $2$. Indeed, for all $z \in N_k(v_i,G\setminus
\{x,y\}) \cap X_i$, we have $z \in N_1(y)$ (by definition of the
$k$-bidismantling order), and thus the $X_i$-configuration $(y,z)$ is
marked with label $1$. Hence, by setting $(x_{(c,v_i)},y_{(c,v_i)})=(x,y)$, the
procedure {\tt Mark($X_i$)} marks $(c,v_i)$ with label $2$.

Assume now that the induction hypothesis holds for some $\ell \geq
1$ and we will show that it still holds for $\ell+1$. Let $(c,r)$ be a
$X_{i+1}$-configuration marked by {\tt Mark($X_{i+1}$)} with label $\ell+1$.
We first prove that $(c,r)$ is eventually marked by {\tt
Mark($X_i$)}. By definition of {\tt Mark($X_{i+1}$)}, there exist
$y_{(c,r)} \in N_1(c) \cap X_{i+1}$ and $x_{(c,r)} \in (N_1(y_{(c,r)})
\setminus \{r\}) \cap X_{i+1}$ such that for all $z \in
N_k(r,G\setminus \{x_{(c,r)},y_{(c,r)}\}) \cap X_{i+1}$, the
$X_{i+1}$-configuration $(y_{(c,r)},z)$ is marked with label at most $\ell$ by
{\tt Mark($X_{i+1}$)}. By the induction hypothesis, this
implies that for all $z \in N_k(r,G\setminus \{x_{(c,r)},y_{(c,r)}\})
\cap X_{i+1}$, the $X_{i+1}$-configuration $(y_{(c,r)},z)$ is
marked by {\tt Mark($X_i$)}. If $v_i \notin N_k(r,G\setminus
\{x_{(c,r)},y_{(c,r)}\})$, then clearly $(c,r)$ is  marked by {\tt
Mark($X_i$)}. Let us assume that $v_i \in N_k(r,G\setminus
\{x_{(c,r)},y_{(c,r)}\})$. We aim at proving that $(y_{(c,r)},v_i)$
is eventually marked by {\tt Mark($X_i$)}. We distinguish three cases.
\begin{itemize}
\item If $y_{(c,r)}=y$, then $(y_{(c,r)},v_i)$ is marked with label $1$
since $y_{(c,r)}=y \in N_1(v_i)$.
\item
If $x_{(c,r)}=y$, then $(y_{(c,r)},v_i)$ is marked with label $1$ or $2$ by
setting $(x_{(y_{(c,r)},v_i)},y_{(y_{(c,r)},v_i)})=(x,y)$. Indeed, for all
$z \in N_k(v_i,G\setminus \{x,y\}) \cap X_i$, we have $z
\in N_1(y)$ (by definition of the $k$-bidismantling order), and thus the
$X_i$-configuration $(y,z)$ is marked with label $1$.
\item
Otherwise, we assert that $(y_{(c,r)},y)$ has already been marked by
{\tt Mark($X_i$)}. By the induction hypothesis, this implies that
$(y_{(c,r)},v_i)$ was also marked.

If $y \in N_k(r,G\setminus
\{x_{(c,r)},y_{(c,r)}\}) \cap X_{i+1}$ and since $(c,r)$ is marked with label
$\ell+1$ by the marking
procedure in $X_{i+1}$, then $(y_{(c,r)},y)$ must be marked by {\tt Mark($X_{i+1}$)}
with label at most $\ell$. By the induction
hypothesis, this implies that $(y_{(c,r)},y)$ has been marked by {\tt
Mark($X_i$)}. Hence, it remains to show that $y \in
N_k(r,G\setminus \{x_{(c,r)},y_{(c,r)}\}) \cap X_{i+1}$.

Let $P$ be a
path of length at most $k$ between $r$ and $v_i$ in $G\setminus
\{x_{(c,r)},y_{(c,r)}\}$. If $x$ or $y$ belongs to $P,$ then we
trivially get that $y \in N_k(r,G\setminus
\{x_{(c,r)},y_{(c,r)}\}) \cap X_{i+1}$. Otherwise, this means that
$r \in N_k(v_i,G\setminus \{x,y\}) \cap X_i$ and $r \in N_1(y)$ holds by
definition of the bidismantling order. Hence, $y \in N_k(r,G\setminus
\{x_{(c,r)},y_{(c,r)}\}) \cap X_{i+1}$.
\end{itemize}
In all three cases, the pair $(y_{(c,r)},v_i)$ is marked by {\tt
Mark($X_i$)}. Thus, for all $z \in N_k(r,G\setminus
\{x_{(c,r)},y_{(c,r)}\}) \cap X_i$, the $X_i$-configuration $(y_{(c,r)},z)$ has been marked.
Therefore, this is also the case for the $X_i$-configuration $(c,r)$.
To conclude the proof, we need to show that, once a
$X_i$-configuration $(c,y)$ ($c \neq v_i$) is marked by {\tt
Mark($X_i$)}, then $(c,v_i)$ can be marked as well. Since $(c,y)$ has been marked,
there exist $y_{(c,y)} \in N_1(c) \cap X_i$ and $x_{(c,y)} \in
(N_1(y_{(c,y)}) \setminus \{y\}) \cap X_i$ such that for all $z \in
N_k(y,G\setminus \{x_{(c,y)},y_{(c,y)}\}) \cap X_i$, the
$X_i$-configuration $(y_{(c,y)},z)$ is marked.  Let
$z' \in N_k(v_i,G\setminus \{x_{(c,y)},y_{(c,y)}\}) \cap X_i$. We prove that
$z' \in N_k(y,G\setminus \{x_{(c,y)},y_{(c,y)}\}) \cap
X_i,$ which shows that $(y_{(c,y)},z')$ has been already marked.
Let $P$ be a shortest path between $v_i$ and $z'$ in $G\setminus
\{x_{(c,y)},y_{(c,y)}\}$. Note that $|P|\leq k$. If $y \in P$, clearly $z' \in N_k(y,G\setminus \{x_{(c,y)},y_{(c,y)}\}) \cap
X_i$. Else, if $x \in P$, then let $P'$ be the subpath of $P$ from $z'$ to
$x$. Then $P' \cup \{x,y\}$ is a path of length at most $k$ between $z'$ and $y$ in the graph $G\setminus
\{x_{(c,y)},y_{(c,y)}\}$. Otherwise, $z' \in N_k(v_i, G \setminus \{x,y\}) \cap X_i$ and thus $z' \in N_1(y)$.
Therefore, for any $z' \in N_k(v_i,G\setminus \{x_{(c,y)},y_{(c,y)}\})
\cap X_i$, $(y_{(c,y)},z')$ is marked and thus the pair
$(c,v_i)$ can be marked as well.

Summarizing, we conclude that for all $c,r \in X_{i+1}$, the configurations $(c,r)$ and
$(c,v_i)$ are marked by the procedure {\tt
Mark($X_i$)}. To
conclude the proof, note that any configuration $(v_i,r)$ can be marked as well: either
with $1$ if $r \in N_1(v_i)$ or by setting
$(x_{(v_i,r)},y_{(v_i,r)})=(y,y)$ otherwise.
\end{proof}

From Theorem~\ref{theor:CWW} and by noticing that if a graph $G=(V,E)$
with $n$ vertices is $n$-bidismantlable, then there are two vertices $x,y$ such that $y$
dominates a connected component of $G\setminus \{x,y\}$, we obtain the following
observation:

\begin{proposition} \label{allkwitness}
${\mathcal C}{\mathcal W}{\mathcal W}$ is the class of graphs which are $k$-bidismantlable  for all $k \geq 1$.
\end{proposition}

We continue with an example showing that  $k$-bidismantlability is not a
necessary condition for any $k\ge 3$.

\begin{figure}[h!]
\begin{center}
{\input{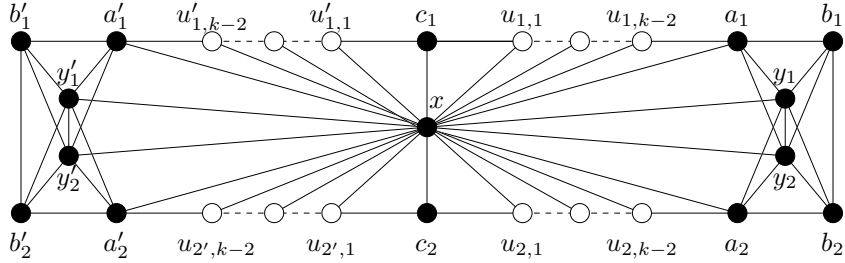}}
\caption{A graph $G \in \CWW(k)$ that is not $k$-bidismantlable}
\label{fig-contre-ex-k}
\end{center}
\end{figure}

\begin{proposition}
Let $G$ be the graph from Fig. \ref{fig-contre-ex-k}. Then $G \in \CWW(k)$,
however $G$ is not k-bidismantlable.
\end{proposition}

\begin{proof}
To show that $G \in \CWW(k)$, we exhibit a strategy for $\cC$. The cop
starts at the vertex $x$. To avoid being captured immediately, the
robber starts in $b_1$, $b_2$, $b_1'$, or $b_2',$ say in $b_1$ (the other cases are
similar). Then the cop moves to $a_1$, goes to $x$ and stays there
during $k-2$ steps, and finally goes to $y_1$. Once $\cC$ is in $a_1$,
$\cR$ can go to $b_1,b_2,y_1$ or $y_2$. Then, while $\cC$ is in $x$,
$\cR$ can move to a vertex in $\{a_1,a_2,b_1,b_2,y_1,y_2\} \cup
\{u_{1,i},u_{2,i}: 2\le i\le k-2\}$. Finally, when the cop moves to
$y_1$, $\cR$ can go to a vertex in $\{a_1,a_2,b_1,b_2,y_2,x\} \cup
\{u_{1,i},u_{2,i}: i\le k-2\}$. Thus, if the robber is not
catched the second time he shows up, he is in a vertex of
$\{u_{1,i},u_{2,i}:  i\le k-2\}$. If the robber shows up in a
vertex of $\{u_{1,i}:  i\le k-2\}$ while the cop is in $y_1$
(the other case is similar), the cop oscillates between $a_1$ and $x$
during $k$ steps and finishes in $x$ (if $k$ is odd, $\cC$ first moves
to $x$ and if $k$ is even, $\cC$ first moves to $a_1$). Thus, the next
time $\cR$ shows up, he is in a vertex of
$\{c_1,a_1,a'_1\}\cup\{u_{1,i},u_{1,i}':  i\le k-2\}$: in any case,
$\cC$ (positioned at $x$) catches $\cR$ when he shows up.

Now we show that $G$ is not $k$-bidismantlable.  We can eliminate a
vertex $v$ if there exist two neighbors $x(v),y(v)$ such that
$N_k(v,G\setminus \{x,y\}) \cap X \subseteq N_1(y),$ where $X$ is the
set of vertices that have not been yet eliminated.  First note that
for any $ i\le k-2$, we can eliminate $u_{1,i}$ with $x(u_{1,i}) =
a_{1,i}$ and $y(u_{1,i}) = x$. By symmetry, we can also eliminate
$u_{2,i}$, $u'_{1,i}$ and $u'_{2,i}$ for any $ i\le k-2$ (these
vertices are colored white in Fig.~\ref{fig-contre-ex-k}). We show
that no other vertex can be eliminated after a set $Y \subseteq
\{u_{1,i}, u_{1,i}', u_{2,i}, u_{2,i}': i\le k-2\}$ of vertices has
been eliminated ($Y$ can be empty or contain all these vertices).

Let $X = V(G) \setminus\{u_{1,i},u_{2,i},u_{1,i}',u_{2,i}':  i\le k-2\}$. By symmetry,
it is sufficient to show that for any $v \in
\{a_1,b_1,c_1,x,y_1\}$, and any adjacent vertices $x(v), y(v) \in
V(G)$, there exists $z(v) \in (N_k(v,G\setminus \{x(v),y(v)\}) \cap X)
\setminus N_1(y(v))$. For any $v$ and $y(v)$, this condition is true as
soon as $N_1(v) \cap X \not\subseteq N_1(y(v))$.

If $v = x$, then there does not exists any $y(v)$ such that $N_1(v)
\cap X \subseteq N_1(y(v))$. If $v = a_1$ and $N_1(v) \cap X \subseteq N_1(y(v)),$
then $y(v) \in\{y_1,y_2\}$; in both cases, $x(v) \notin \{u_{1,i}:  i\le k-2\}$
and thus $z(v) = c_1$ satisfies the condition. If $v = b_1$ and $N_1(v) \cap X \subseteq N_1(y(v))$ then $y(v) \in
\{y_1,y_2\}$; by symmetry, we assume $y(b_1) = y_1$. Note that there
exist two vertex-disjoint paths $(b_1,a_1,u_{1,k-2},\ldots,u_{1,1},c_1)$ and
$(b_1,y_2,x,c_1)$ of length at most $k$ from $b_1$ to
$c_1$ avoiding $y_1.$ Consequently, for any choice of $x(b_1)$, the vertex $z(v) =
c_1$ satisfies the condition. If $v = y_1$ and $N_1(v) \cap X \subseteq N_1(y(v)),$
then $y(v)=y_2$. Again, there exist two vertex-disjoint paths $(y_1,a_1,u_{1,k-2},\ldots,u_{1,1},c_1)$ and
$(y_1,x,c_1)$ of length at most
$k$ from $y_1$ to $c_1$ avoiding $y_2.$ Hence, for any choice of $x(y_1)$, the vertex $z(v) = c_1$
satisfies the condition. Finally, suppose that $v = c_1$ and $N_1(v) \cap X \subseteq N_1(y(v)).$ Then $y(v) \in
\{u_{1,1},u_{1,1}',x\}$. If $y(v) = u_{1,1}$ (the case $y(v) =
u_{1,1}'$ is similar), then $x(v) \notin \{a'_1\}\cup\{u_{1,i}':  i\le k-2\}$ and thus $z(v) = b_1$ satisfies the condition. If $y(v) =
x$, then either $x(v) \notin \{a'_1\}\cup\{u_{1,i}':  i\le k-2\}$,
or $x(v) \notin \{a_1\}\cup\{u_{1,i}:  i\le k-2\}$. By symmetry, we
assume $x(v) \notin \{a'_1\}\cup\{u_{1,i}':  i\le k-2\}$; in this
case, $z(v) = b_1$ satisfies the condition.
\end{proof}

\medskip\noindent
{\bf Open question 3:} Characterize the $k$-winnable graphs for
$k=2,3$ and, more generally, for all $k$.

\subsection{Cop-win graphs for game with fast robber and witness.}

We now consider a variant of the game where the robber is visible
every $k$ moves and has speed $s$ while the cop has speed $1$. It
means that at each step, the robber can move to a vertex at distance
at most $s$ from his current position, and that the cop can see the
robber only every $k$ steps.  We denote by $\CWFRW(s,k)$ the class of
graphs where a single cop with speed $1$ can catch a robber with speed
$s$ that is visible every $k$ moves. By definition, we have
$\CWFRW(1,k) = \CWW(k)$ and $\CWFRW(s,1) = \CWFR(s)$.

\begin{theorem}\label{th:CWFRW}
If $s \geq 3, k \geq 1$ or $s \geq 2, k \geq 2$, then $\CWFRW(s,k)$ is
the class of big brother graphs.
\end{theorem}

\begin{proof}
We know from Theorem~\ref{th-speed-3} that if $s \geq 3$ and $k \geq
1$, $\CWFR(sk) = \CWFR(s)$ is the class of big brother
graphs. Consequently, since $\CWFR(sk) \subseteq \CWFRW(s,k) \subseteq
\CWFR(s)$, it follows that $\CWFRW(s,k)$ is the class of big brother graphs for all $s
\geq 3$ and $k \geq 1$.

In the remaining of this proof, we show that when $s = 2$ and $k \geq
2$, $\CWFRW(2,k)$ also coincides with the class of big brother graphs.
This proof follows closely the proof of Theorem~\ref{theor:CWW}. In
particular, the following proposition is the counterpart of
Proposition~\ref{prop:particularEdge}.

\begin{proposition}\label{prop:CWFRW-y}
Let $G \in \CWFRW(2,k)$ for $k \geq 2$. If the minimum degree of $G$
is at least $2$, then $G$ contains two vertices $v$ and $y$ such that
$N_{2k}(v,G\setminus\{y\}) \subseteq N_1(y)$.
\end{proposition}

\begin{proof}
If $G$ contains a dominating vertex $y$, then the result holds for
any $v \neq y$. Assume thus that $G$ has no dominating vertices.
Consider a parsimonious winning strategy of the cop and suppose that
the robber uses a strategy to avoid being captured as long as
possible.  Since $G$ does not contain leaves, the robber is caught
immediately after having been visible. Since $G$ does not have any
dominating vertex, the robber is visible at least twice. Let $y$ be
the vertex occupied by the cop when the robber becomes visible for
the last time before his capture. Let $v$ be the next-to-last
visible vertex occupied by the robber.  Finally, let
$S_c=(c_0,c_1,\ldots,c_{k}=y)$ be the trajectory of the cop during
the last $k$ steps (repetitions are allowed). Note that $v \notin
N_1(c_0),$ otherwise the robber would have been caught immediately.
We distinguish two cases depending of whether or not $c_i = y$ for
all $1 \leq i \leq k$.

If for any $1 \leq i \leq k$, we have $c_i = y$, then $\cR$ could have
move to any vertex $w \in N_{2k}(v,G\setminus\{y\})$. Since the
trajectory of the robber is maximal, the robber is caught in any such
$w$ and thus $N_{2k}(v,G\setminus\{y\}) \subseteq N_1(y)$. Suppose now
that there exists $i$ such that $c_i \neq y$. Let $i$ be
the largest index such that $c_i \neq y$.

\vspace{-0.5ex}
\begin{claim}\label{claim-cycle-CWFRW}
If $G$ contains a cycle $C$ and a vertex $w\in C$ such that $d(v,w)< d(c_1,w),$
then $G\setminus\{y\}$ has a connected component that is dominated by $y$.
\end{claim}
\vspace{-0.5ex}

\begin{proof}
Let $w$ be the closest to $v$ vertex satisfying the condition of the
claim. If the assertion of the claim is not satisfied, we will exhibit
a strategy allowing the robber to escape the cop during more steps,
contradicting the choice of the trajectory $\cR$.  Suppose that at the
beginning of the last phase, the robber moves from $v$ to $w$ along a
shortest $(v,w)$-path.  Since $d(v,w)<d(c_1,w)$, the robber cannot be
intercepted by the cop during these moves.  Suppose that the robber
reaches the vertex $w$ before the $i$th step when the cop arrives at
$c_i.$ Then by Lemma \ref{cycle} (adapted to this game) the robber can
safely move on $C$ until the cop reaches the vertex $c_i$.

In both cases, let $z$ be the current position of the robber when the
cop reaches $c_i$.  Then $z \in N_1(y),$ otherwise the robber can remain
at $z$ without being caught because starting with this
step the cop remains in $y$. If $z \neq y$, then let $u = z$, and if
$z = y$, let $u$ be a neighbor of $y$ different from $c_i$ (it exists
because the minimum degree of a vertex of $G$ is at least $2$). In
both cases, let $H$ be the connected component of $G\setminus\{y\}$
that contains $u$.  We assert that $y$ dominates all the vertices of
$H$. Suppose this is not the case and consider a vertex $t$ in
$V(H)\setminus N(y)$ that is at a minimum distance from $u$ in
$H$. From our choice of $t$, we can find a common neighbor $r \in
V(H)$ of $y$ and $t$.  If $r \neq c_i$, then while the cop is in
$c_i$, the robber can go from $z$ ($z$ is either $u$ or $y$) to $r$
through $y$ and then, when $\cC$ goes to $y$, $\cR$ goes to $t$ and
stays there until he becomes visible. If $r = c_i$, then let $s$ be a
neighbor of $r$ on a shortest $(u,r)$-path in $G\setminus\{y\}$. By
our choice of $t$, necessarily $s \in N(y)$. Thus, when the cop is in $c_i$, the
robber can go from $z$ to $s$ through $y$. And then, when the cop goes
to $y$, $\cR$ goes to $t$ through $r$ and stays there until he becomes
visible. In both cases, by following such a strategy, $\cR$ could avoid
being caught. This contradicts the maximality of the trajectory of the
robber. This concludes the proof of the claim.
\end{proof}

We now complete the proof of Proposition~\ref{prop:CWFRW-y}. If the
vertex $v$ belongs to a cycle $C$, then setting $w:=v$ and applying
Claim~\ref{claim-cycle-CWFRW}, we conclude that $y$ dominates a
non-empty connected component $H$ of $G\setminus\{y\}$ establishing
thus the assertion.  So, suppose that $v$ is an articulation point of
$G$ not contained in a cycle. Since the minimum degree of $G$ is at least $2$, $G\setminus \{
v\}$ has a connected component $H$ that does not contain $c_0$ (nor
$c_1$). Necessarily $H$ contains a cycle $C$, otherwise we will find
in $H$ a vertex of degree $1$ in $G$. Since any path from $c_1$ to a
vertex $w$ of $C$ goes through $v$, we obtain $d(v,w)<d(c_1,w)$.
Then, the result again follows from
Claim~\ref{claim-cycle-CWFRW}. This ends the proof of
Proposition~\ref{prop:CWFRW-y}.
\end{proof}

Finally, we prove Theorem~\ref{th:CWFRW} when $s = 2$ and $k \geq2$.
Consider a graph $G \in \CWFRW(2,k)$. To establish that $G$ is a big brother graph, in view of
Theorem~\ref{th-speed-3}, it suffices to show that $G$ is
$(2k,1)$-dismantlable. For this, by Proposition~\ref{prop-sans-X}, we
just have to show that there exists an ordering $v_1, \ldots, v_n$ of
the vertices of $G$ such that for each $1 \leq i < n$ there exists $y
\in \{v_{i+1}, \ldots, v_n \}$ such that
$N_{2k}(v_i,G_i\setminus\{y\}) \subseteq N_1(y,G_i)$.  We proceed by
induction on the size of $G$. Suppose that $G$ has at least two
vertices, otherwise the result is trivial.  If $G$ has a vertex $v$ of
degree $1$, then let $y$ be the unique neighbor of $v$ in $G$. In this
case, then obviously $N_{2k}(v,G\setminus\{y\}) \subseteq
N_1(y,G)$. Otherwise, by Proposition~\ref{prop:CWFRW-y}, we can find
vertices $v$ and $y$ such that $N_{2k}(v,G\setminus\{y\}) \subseteq
N_1(y,G)$.

We now show that $G' = G\setminus \{v_1\}$ also belongs to
$\CWFRW(2,k)$. Consider a winning positional strategy $\sigma$ for the
cop in $G$. As in the proof of Theorem~\ref{copwin}, we define a
strategy $\sigma'$ for the cop in $G'$ using one bit of
memory. Starting from $\sigma$, we define $\sigma'(c,r,m)$ for any
positions $c, r \in V(G')$ of the cop and the robber and for any value
of the memory $m \in \{0,1\}$.  The idea is that the cop plays using
$\sigma$ except when he is in $y$ and his memory contains $1$; in this
case he uses $\sigma$ as if he was in $v$ (going via $y$ instead
of $v$ if $v$ appears in his sequence of moves).

If $m = 0$ or $c \neq y$, let $\sigma(c,r) = (c_1,\ldots, c_k)$ and
for each $i$, let $c'_i = c_i$ if $c_i\neq v$ and $c'_i = y$ otherwise
(this is possible since $N_1(v)\subseteq N_1(y)$). If $c_k = v$, then
$\sigma'(c,r,m) = ((c'_1, \ldots, c'_{k-1},y),1)$, otherwise let
$\sigma'(c,r,m) = ((c'_1, \ldots, c'_{k-1},c_k),0)$.

If $m = 1$ and $c = y$, let $\sigma(v,r) = (c_1,\ldots, c_k)$ and for
each $i$, let $c'_i = c_i$ if $c_i\neq v$ and $c'_i = y$ otherwise. If
$c_k = v$, then $\sigma'(y,r,1) = ((c'_1, \ldots, c'_{k-1},y),1)$,
otherwise let $\sigma'(y,r,1) = ((c'_1, \ldots, c'_{k-1},c_k),0)$.

Let $S_r = (r_1, r_2, \ldots, r_p, \ldots)$ be a valid sequence of
moves in $G'$. Since $V(G')\subseteq V(G)$, $S_r$ is also a valid
sequence of moves in $G$.  Let $S_c =(c_1, \ldots, c_p, \ldots)$ be
the corresponding valid sequence of moves of the cop playing
$\sigma$ against $S_r$ in $G$ and let $S'_c=(c'_1, \ldots, c'_p,
\ldots)$ be the valid sequence of moves of the cop playing $\sigma'$
against $S_r$ in $G'$. Note that the sequences of moves $S_c$ and
$S_c'$ differ only if $c_k = v$ and $c_k' = y$.  Finally, since the
cop follows a winning strategy for $G$, there exists a step $j$ such
that $c_j=r_j \in V(G')$ (note that $r_j\neq v$ because we supposed
that $S_r$ is a valid sequence of moves in $G'$).  Since $c_j \neq
v$, we also have $c'_j=r_j$, thus $\cC$ captures $\cR$ in the game
restricted to $G'$. In conclusion, starting from a positional
strategy for the game in $G$, we have constructed a winning strategy
using memory for the game in $G'$. As mentioned in the introduction,
it implies that there exists a positional winning strategy for the
game in $G'$. Consequently, $G' \in \CWFRW(2,k)$ and by induction
hypothesis, $G'$ has $(2k,1)$-dismantling order $(v_2,\ldots, v_n)$,
whence $(v,v_2,\ldots,v_n)$ is a $(2k,1)$-dismantling order of $G$.
This concludes the proof of Theorem~\ref{th:CWFRW}.
\end{proof}

\section{Bipartite cop-win graphs for game with ``radius of capture''}

In this section we characterize bipartite graphs of the class
${\mathcal C}{\mathcal W}{\mathcal R}{\mathcal C}(1)$, i.e., the
bipartite cop-win graphs in the cop and robber game with radius of
capture $1$. Recall that in this game introduced in~\cite{BonChi},
$\mathcal C$ and $\mathcal R$ move at unit speed and the cop wins if
after his move he is within distance at most $1$ from the
robber. Notice that any graph of diameter 2 belongs to ${\mathcal
  C}{\mathcal W}{\mathcal R}{\mathcal C}(1):$ given the positions $u$
and $v$ of the cop and the robber, to capture the robber the cop
simply moves from $u$ to a common neighbor of $u$ and $v$.

Following \cite{BaFaHe}, a bipartite graph $G$ is called {\it
  dismantlable} if its vertices can be ordered $v_1,\ldots,v_n$ so
that $v_{n-1}v_n$ is an edge of $G$ and for each $v_i, i < n-1$, there
exists a vertex $y:=v_j$ with $j>i$ (necessarily not adjacent to
$v_i$) such that $N(v_i,G_i):=N_1(v_i,G_i)\setminus \{ v_i\}\subseteq
N_1(y)$. Note that for any $i$, $G_i$ is a retract of $G$ and
therefore an isometric subgraph of $G$.

\begin{theorem} \label{bipartite} A bipartite graph $G$ belongs to
  ${\mathcal C}{\mathcal W}{\mathcal R}{\mathcal C}(1)$ if and only if
  $G$ is dismantlable.
\end{theorem}

\begin{proof}
First suppose that $G\in {\mathcal C}{\mathcal W}{\mathcal R}{\mathcal
  C}(1)$. If $G$ has diameter $2$, then necessarily $G$ is a complete
bipartite graph, which is obviously dismantlable. Suppose now that $G$
has diameter at least $3$. As in previous proofs of similar results,
we assume that $\cC$ uses a parsimonious strategy.
Consider a maximal sequence of moves of the robber before he
get caught.  Let $v$ be the next-to-last position of the robber and
let $y$ be the position of the cop at this step ($y$ is not adjacent
to $v$, otherwise $\mathcal C$ captures $\mathcal R$ in $v$). This
means that for any $w \in N_1(v)$, the cop can move in some vertex
$u \in N_1(y)$ such that $w \in N_1(u)$. This shows that
$N_1(v)\subseteq N_2(y)$. Since $G$ is bipartite, this means that
$d(v,y)=2$ and all neighbors of $v$ are adjacent to $y$, i.e.,
$N(v) \subseteq N(y)$.

We now show that $G' = G\setminus \{v\}$ also belongs to
$\CWRC(1)$. Consider a winning positional strategy $\sigma$ for the
cop in $G$. As in the proof of Theorem~\ref{copwin}, we define a
strategy $\sigma'$ for the cop in $G'$ using one bit of
memory. Starting from $\sigma$, we define $\sigma'(c,r,m)$ for any
positions $c, r \in V(G')$ of the cop and the robber and for any value
of the memory $m \in \{0,1\}$.  The idea is that the cop plays using
$\sigma$ except when he is in $y$ and his memory contains $1$; in this
case he uses $\sigma$ as if he was in $v$ (going to $y$ instead of
$v$).  If $m = 0$ or $c \neq y$, if $\sigma(c,r) = v$ then
$\sigma'(c,r,m) := (y,1)$ (this is a valid move since $N(v) \subseteq
N(y)$ and $c \neq v$), otherwise let $\sigma'(c,r,m) :=
(\sigma(c,r),0)$.  If $m = 1$ and $c = y$, if $\sigma(v,r) = v$, then
$\sigma'(y,r,1) := (y,1)$, otherwise let $\sigma'(y,r,1) :=
(\sigma(v,r),0)$.

Let $S_r = (r_1, r_2, \ldots, r_p, \ldots)$ be a valid sequence of
moves in $G'$. Since $V(G')\subseteq V(G)$, $S_r$ is also a valid
sequence of moves in $G$.  Let $S_c =(c_1, \ldots, c_p, \ldots)$ be
the corresponding valid sequence of moves of the cop playing $\sigma$
against $S_r$ in $G$ and let $S'_c=(c'_1, \ldots, c'_p, \ldots)$ be
the valid sequence of moves of the cop playing $\sigma'$ against $S_r$
in $G'$. Note that the sequences of moves $S_c$ and $S_c'$ differ only
if $c_k = v$ and $c_k' = y$.  Finally, since the cop follows a winning
strategy for $G$, there exists a step $j$ such that $c_{j+1} \in N_1(r_j)
\subseteq V(G')$ (note that $r_j\neq v$ because we supposed that $S_r$ is a
valid sequence of moves for the game in $G'$).  Since $N(v) \subseteq N(y)$, we also have
$c'_{j+1} \in N_1(r_j)$, thus $\cC$ captures $\cR$ in the game restricted
to $G'$. In conclusion, starting from a positional strategy for the
game in $G$, we have constructed a winning strategy using memory for
the game in $G'$. As mentioned in the introduction, it implies that
there exists a positional winning strategy for the game in $G'$.
Consequently, $G' \in \CWRC(1)$ and by induction hypothesis, $G'$ has
a dismantling order $(v_2,\ldots, v_n)$, whence $(v,v_2,\ldots,v_n)$
is a dismantling order of $G$.

Conversely, suppose that a bipartite graph $G$ is dismantlable and let
$v=v_1,v_2,\ldots, v_n$ be a dismantling order of $G$. If $G$ has
diameter $2$, then $G$ is a complete bipartite graph and thus, $G \in
\CWRC(1)$. Suppose now that $G$ has a diameter at least $3$.  By
induction hypothesis, $G'=G(\{ v_2,\ldots v_n\})$ belongs to
${\mathcal C}{\mathcal W}{\mathcal R}{\mathcal C}(1)$.  Suppose that
$v$ is dominated by a vertex $y$ at distance $2$ from $v$.

Consider a parsimonious positional winning strategy $\sigma'$ for the
cop in $G'$. Using $\sigma'$, we build a parsimonious positional
winning strategy $\sigma$ for the cop in $G$. As in the previous
proofs, the idea is that if $\cC$ sees $\cR$ in $v$, he plays as in
the game on $G'$ when the cop is in $y$. For any positions $c, r \in
V(G)$, if $r \in N_2(c)$, then $\sigma(c,r) := u \in N_1(c)\cap
N_1(r)$. Otherwise, if $c\in N(y)\setminus N(v)$ and $r = v$, then
$\sigma(c,v) := y$.  Otherwise, if $c,r\neq v$, then
$\sigma(c,r) := \sigma'(c,r)$; if $r = v$ and $c \notin N_2(v)$,
then $\sigma(c,v) := \sigma'(c,y)$; finally, if $c = v$ and $r \notin N_2(v)$, then
$\sigma(v,r) := u \in N(v)$ (in fact, if the cop plays according to
$\sigma$, he will never move to $v$). By
construction, $\sigma$ is parsimonious and positional; in particular,
$\sigma(y,v) \in N(v)$.

We now prove that $\sigma$ is a winning strategy. Consider any valid
sequence $S_r=(r_1, \ldots, r_k, \ldots)$ of moves of the robber in
$G$, and let $S'_r=(r_1', \ldots, r_k', \ldots)$ be the sequence
obtained by setting $r'_k=r_k$ if $r_k \neq v$ and $r'_k=y$ if
$r_k=v$. Since $N(v)\subseteq N(y)$, $S'_r$ is a valid sequence of
moves for the robber in $G'$. By induction hypothesis, for any initial
position of $\cC$ in $G'$, the strategy $\sigma'$ enables $\cC$,
following a trajectory $S_c' = (c_1', \ldots, c_k',\ldots)$, to catch
$\cR$ which moves according to $S_r'$, i.e., there exists an index $m$
such that $r'_{m}\in N_1(c_{m+1}')$.  Suppose that $\mathcal C$
chooses his starting position $c_1$ in $G'$ and that the cop,
following $\sigma$, plays in $G$ the sequence $S_c = (c_1, \ldots,
c_k,\ldots)$ against the sequence $S_r$ of the robber. Note that $c_k
= c_{k}'$ for any $k < m$, i.e., except when the robber is caught in
$G'$. If $r_m=r'_m$, then $r_m\in N_1(c_{m+1})$ and $\mathcal C$
captures $\mathcal R$. Now suppose that $r_m\neq r'_m$.  From the
definition of $S'_r$ we conclude that $r'_m=y$ and $r_m=v$. If $v \in
N(c_{m+1})$, the robber is caught at step $m+1$. Otherwise, since
$c_{m+1} \in N(y)\setminus N(v)$, then $c_{m+2} = y$. If at step $r_{m+1}$ the robber
moves to a neighbor of $v$, since $N(v) \subseteq
N(y)$, we conclude that $r_{m+1}\in N(c_{m+2})$ and the robber is
caught. Finally, if the robber remains at $v$ (i.e., $r_{m+1}=v$),
then to avoid being caught while moving, at step $m+2$, $\cR$ must
also stay in $v$ and then at step $m+3$ the cop moves from $y$ to any
common neighbor of $y$ and $v$  and catches the robber.
\end{proof}

In the classical game of cop and robber where both players have
speed $1$ and there is no radius of capture, a bipartite graph $G$
is cop-win if and only if $G$ is a tree. The previous result shows
that in the variant of the game with a radius of capture, a single
cop can win in a considerably  larger class of graphs.

Bonato and Chiniforooshan~\cite{BonChi} asked for a characterization
of graphs of $\CWRC(1)$. Theorem~\ref{bipartite} answers the question
in the case of bipartite graphs. On the other hand, characterizing the graphs of
$\CWRC(1)$ using a specific dismantling scheme seems to be quite
challenging. Two natural candidates for us were the total orders $v_1, \ldots,
v_n$ satisfying the following conditions:
\begin{enumerate}[(i)]
\item for each vertex $v_i$, there exists a vertex $v_j,j > i,$ such that $N_1(v_i,G_i)
  \subseteq N_2(v_j,G_{i+1})$;
\item for each vertex $v_i$, there exists a vertex $v_j,j > i,$ such that $N_2(v_i,G_i)
  \subseteq N_2(v_j,G_{i})$.
\end{enumerate}

It seems that the first condition is necessary, while the second condition is
sufficient. However, we were not able to prove this. In fact, one can
easily show that any graph $G \in \CWRC(1)$ contains two vertices $v,
y$ such that $N_1(v,G_i) \subseteq N_2(y,G_{i+1})$, but we cannot show
that $G\setminus\{v\}$ also belongs to $\CWRC(1)$. A similar
difficulty occurs while establishing the sufficiency of the second
dismantling order.

\section*{Acknowledgements}

\noindent
Work of J. Chalopin, V. Chepoi, and Y. Vax\`es
was supported in part by the ANR grants SHAMAN (ANR VERSO) and OPTICOMB (ANR BLAN06-1-138894).
Work of N. Nisse was supported by the ANR AGAPE and DIMAGREEN, and the European project IST FET AEOLUS.


\end{document}